\documentclass[journal]{IEEEtran}

\usepackage{amsfonts}
\usepackage{amsmath}
\usepackage{amssymb}
\usepackage{amsthm}
\usepackage{lscape}
\usepackage{epsf}
\usepackage{subfig}
\usepackage{graphicx}
\usepackage{verbatim}

\newtheorem{theorem}{\indent Theorem}[section]
\newtheorem{lemma}[theorem]{\indent Lemma}
\newtheorem{corollary}[theorem]{\indent Corollary}

\newtheorem{EXAMPLE}{\indent Example}[section]
\newtheorem{remark}[theorem]{\indent Remark}
\newtheorem{definition}[theorem]{\indent Definition}

\newenvironment{example}{\begin{EXAMPLE}\rm}{\rm\end{EXAMPLE}}

\newcommand{\code}{{\mathcal{C}}}

\newcommand{\cL}{{\mathcal{L}}}

\newcommand{\cH}{{\mathcal{H}}}
 
\newcommand{\cE}{{\mathcal{E}}}
\newcommand{\decoder}{{\mathcal{D}}}

\newcommand{\cB}{{\mathcal{B}}}
\newcommand{\cA}{{\mathcal{A}}}
\newcommand{\cD}{{\mathcal{D}}}

\newcommand{\cK}{{\mathcal{K}}}
\newcommand{\cP}{{\mathcal{P}}}

\newcommand{\cS}{{\mathcal{S}}}

\newcommand{\csf}{{\mathsf{c}}}

\newcommand{\distance}{{\mathsf{d}}}
\newcommand{\Distance}{{\mathsf{D}}}

\newcommand{\cc}{{\mathsf{C}}}

\newcommand{\Code}{{\mathbb{C}}}

\newcommand{\ff}{{\mathbb{F}}}

\newcommand\nn{{\mathbb N}}

\newcommand{\bldA}{{\mbox{\boldmath $A$}}}

\newcommand{\blde}{{\mbox{\boldmath $e$}}}

\newcommand{\bldbeta}{{\mbox{\boldmath $\beta$}}}

\newcommand{\bldG}{{\mbox{\boldmath $G$}}}

\newcommand{\bldM}{{\mbox{\boldmath $M$}}}

\newcommand{\bldv}{{\mbox{\boldmath $v$}}}
\newcommand{\bldu}{{\mbox{\boldmath $u$}}}
\newcommand{\bldw}{{\mbox{\boldmath $w$}}}
\newcommand{\bldx}{{\mbox{\boldmath $x$}}}

\newcommand{\bldz}{{\mbox{\boldmath $z$}}}
\newcommand{\bldgamma}{{\mbox{\boldmath $\gamma$}}}

\newcommand{\bldzero}{{\mbox{\boldmath $0$}}}

\newcommand{\define}{\stackrel{\mbox{\tiny $\triangle$}}{=}}
\newcommand{\equival}{\stackrel{\mbox{$\cdot$}}{=}}

\newlength{\Initlabel}
\newlength{\Algwidth}

\title{Hybrid Noncoherent Network Coding
\thanks{%
    V. Skachek is with the Department of Electrical and Computer Engineering, 
    McGill University, 3480 University Street, Montr\'eal, QC H3A 2A7, Canada. 
    This work was done while he
    was with the Coordinated Science Laboratory, University of Illinois at Urbana-Champaign, 1308 W. Main Street, Urbana, IL 61801, USA.}      
\thanks{%
    O. Milenkovic and A. Nedi\'c are with the Coordinated Science Laboratory, University of Illinois at Urbana-Champaign, 
    1308 W. Main Street, Urbana, IL 61801, USA.}  
\thanks{%
    This work is supported by the Air Force Office for Scientific Research and NSF Grants CCF 0939370 and CCF 1117980.}
}  

\author{Vitaly Skachek, Olgica Milenkovic, and Angelia Nedi\'c  }
\date{}    
    
\begin{document}


\maketitle

\begin{abstract}

We describe a novel extension of subspace codes for noncoherent networks, 
suitable for use when the network is viewed as a communication system that 
introduces both dimension and symbol errors.  We show that when symbol erasures occur in a significantly large number of 
different basis vectors transmitted through the network and when the min-cut of the network is much smaller then the length of 
the transmitted codewords, the new family of codes outperforms their subspace code counterparts. 

For the proposed coding scheme, termed hybrid network coding, we derive two upper bounds on the size of the codes. 
These bounds represent a variation of the Singleton and of the sphere-packing bound. We show that a simple concatenated scheme 
that consists of subspace codes and Reed-Solomon codes is asymptotically 
optimal with respect to the Singleton bound. Finally, we describe two efficient decoding algorithms for concatenated subspace codes that
in certain cases have smaller complexity than their subspace decoder counterparts. 
\end{abstract}

\section{Introduction}
Network coding is a scheme introduced by Ahlswede \emph{et al.}~\cite{Ahlswede} for efficient communication over networks with transmission
bottlenecks. The authors of~\cite{Ahlswede} showed that under a broadcast scenario in networks, the maximal theoretically achievable 
communication rate -- called the \emph{capacity} of the network -- can be characterized by 
minimal cuts in the network and achieved by appropriate coding methods. 

In the last decade, network coding became a focal point of research in coding theory. 
There exists a variety of network coding solutions currently used in practice: random 
network coding approach was first proposed in~\cite{Ho}; algebraic coding was shown to achieve 
the capacity of a class of networks in~\cite{Koetter-Medard, Li};  non-linear approaches were also studied in~\cite{Zeger}.

The use of network coding for error-correction was first proposed in~\cite{Cai}.   
When the network topology is not known, or when it changes with time, it was suggested in~\cite{KK} 
to use subspace coding for joint error-correction and network coding, and suitable codes were constructed 
therein. Subspace codes are closely related to the rank-metric codes, also extensively studied in the codning literature~\cite{Gabidulin, Roth, Gadouleau}. 
The parameters of the codes in~\cite{KK} were further improved in a series of subsequent works, including~\cite{Etzion-Silberstein},~\cite{Bossert},~\cite{Kohnert},~\cite{Felice},~\cite{Skachek},~\cite{Trautmann}. 
Bounds on the parameters of subspace codes were derived in~\cite{Etzion-Vardy} and~\cite{Xia}.
It should also be mentioned that the subspace codes, which were proposed in~\cite{KK} for noncoherent network coding, were studied independently in the area of cryptography
under the name \emph{authentication codes}~\cite{Naini}. 

In our work, we follow the line of research started in~\cite{KK}. More specifically, we consider 
error-correction for a special case of network coding, suitable for practical applications in which the topology of the network in not known or changes with time. This type of scheme is, as already mentioned, 
known as coding for {\bf noncoherent networks}. Currently, the only known approach for noncoherent network coding 
utilizes subspace codes. 

Subspace codes for noncoherent network coding are based on the idea that the transmitted data vectors can be associated 
with linear vector subspaces. Linear network coding does not change the information about the subspaces, since it only allows for linear combining of the transmitted bases
vectors. Hence, if there are no errors, the receiver obtains uncompromised information regarding the transmitted subspace. 
The transmitted subspaces can only be modified within the network through the introduction of errors.  
In order to protect the transmitted information one has to add carefully structured redundancy into the subspace messages. 

In the context of the work~\cite{KK}, the errors are modeled as {\bf dimension gains} and 
{\bf dimension losses}. These notions, although of theoretical value, may appear rather abstract in certain 
networking applications, where packets (symbols or collections of symbols) are subjected to erasures or substitution errors. 
One fundamental question remains: how is one to interpret the notion of dimension gains and losses in terms of symbol errors and erasures, and what kind of errors and erasures constitute 
dimension gains and losses? 

We propose {\bf a hybrid approach} to noncoherent network coding, which attempts to connect the notions of dimension loss 
and gain with those of individual symbol errors and erasures. The crux of our approach is to consider network coding where 
dimension gains and losses, in addition to individual symbol errors and erasures, are all possible. 
This allows us to study the trade-offs between the required overhead in the network layer aimed at correcting dimension gains/losses, 
and the overhead in the physical layer designated to correcting symbol erasures and errors.

Our main result shows that by incorporating symbol error-correcting mechanism into subspace codes, one can increase the number of tolerable dimension gains and losses, without compromising 
the network throughput. Hence, the proposed approach leads to an increase in the overall number of 
correctable errors in the subspace-based scheme akin to~\cite{KK}. 

In order to illustrate our approach, consider the following straightforward example. 
Assume the case of a noncoherently coded network in which arbitrary (unknown) ten symbols are erased from the 
basis vectors representing the message. 
The first question is how many dimension losses should be considered in the model of~\cite{KK}? 
One reasonable way to look at it is to assume the worst-case scenario where each symbol erasure
introduces one dimension loss, and each error introduces a simultaneous dimension loss and gain. 
Consequently, ten symbol erasures would amount to ten dimension losses.
However, if there were an alternative way to correct some of these symbol erasures or errors, 
the effective number of dimension losses and gains may become significantly smaller. 
In the example, correcting five symbol erasures would, in the best case, reduce the burden of 
subspace codes in terms of dimension loss recovery by five dimensions. And, at least at first glance, correcting
symbol errors appears to be a task easier to accomplish than correcting dimension errors.

We therefore pose the following questions: what are the fundamental performance bounds for noncoherent network coding schemes, consequently termed hybrid network codes, capable of correcting symbol erasures 
and errors on one side, and dimension gains and losses on the other side? What is the optimal rate allocation scheme for hybrid network codes with respect to 
dimension losses/gains and symbol errors/erasures?
What is the optimal ratio between the two allocation rates and how can it be achieved practically? How does one efficiently correct errors in this new scheme? 
The work in this paper is aimed at answering these questions. 

There are various potential applications for hybrid network codes~\cite{Katti}. Hybrid codes can be useful in networks where no link-layer error-correction is performed. 
Such networks include sensor networks for which the computational power of ``intermediate nodes'' is not sufficiently large. This prevents error-correction to be 
performed before the errors propagate through the network.
Hybrid codes can also be used in 
networks for which a physical layer packet is very small, the network layer packet consists of many physical layer packets, and the packet 
may be regarded as a single symbol. In this case, if an error in the packet cannot be decoded, a symbol error is declared 
which is subsequently ``transferred'' into a dimension loss/gain.

We would also like to point out that upon publication of the preliminary results on our hybrid network coding, two other interesting directions related to this model were 
proposed in literature. A concatenation of subspace codes and algebraic codes was studied in the context of coding 
for distributed data storage in~\cite{Silberstein-Vish}. 
More specifically, it was suggested in~\cite{Silberstein-Vish} to use concatenation of Gabidulin codes~\cite{Gabidulin} 
with classical maximum distance separable code for correction of adversarial errors in distributed storage systems.
In~\cite{Jaggi}, the authors studied symbol-level error-correction in random codes used over 
both coherent and non-coherent networks. However, in that work, no explicit constructions of codes were presented.

The paper is organized as follows. The notation and prior work are discussed in Section~\ref{sec:prior-work}. 
In the sections that follow, we define hybrid codes that can be used for simultaneous correction 
of dimension losses/gains and symbol erasures in noncoherent networks. 
More specifically, the basic code requirements and parameters are presented and described in Section~\ref{sec:code-L}. 
Two upper bounds on the size of hybrid codes, the Singleton bound and the sphere-packing bound, are presented in Section~\ref{sec:bounds}. 
A straightforward concatenated code construction appears in Section~\ref{sec:construct}. The analysis of code parameters and the comparison 
with known subspace code constructions appear in Section~\ref{sec:construction}. 
The decoding algorithm for the proposed codes 
is presented in Section~\ref{sec:decoding-1}.  
In Section~\ref{sec:decoding-2} we show that the same codes can also be used for simultaneous correction of dimension losses and symbol erasures/errors, and 
state some results analogous to those in Sections~\ref{sec:code-L}-\ref{sec:decoding-1}. 
Finally, we discuss some results related to simultaneous correction of both dimension losses/gains and symbol erasures/errors. 



\section{Notation and Prior Work}
\label{sec:prior-work}

Let $W$ be a vector space over a finite field $\ff_q$, where $q$ is a power of a prime number. 
For a set of vectors $S \subseteq W$, we use $\langle S \rangle$ to denote the linear span of the vectors in $S$. 
We also use the notation $\langle \bldu_1, \bldu_2, \cdots, \bldu_\ell \rangle$
for the liner span of the set of vectors $\{ \bldu_1, \bldu_2, \cdots, \bldu_\ell \}$. 
Let $\nn$ be the set of positive integer numbers. 
We write $\bldzero^m$ to denote the all-zero vector of length $m$, for any $m \in \nn$.
When the value of $m$ is clear from the context, we sometimes write $\bldzero$ rather than $\bldzero^m$. 
We also denote by $\blde_i \define (\underbrace{0,\ldots,0}_{i-1},1,\underbrace{0,\ldots,0}_{n-i}) \in \ff_q^n$ a unity vector which has a one in position $i \in \nn$ and zeros in all other positions. 
The length of the vector will be clear from the context.  

Let $V, U \subseteq W$ be linear subspaces of $W$. 
We use the notation $\dim(V)$ for the dimension of $V$. We denote 
the sum of two subspaces $U$ and $V$ 
as $U+V = \{ \bldu + \bldv \; : \; \bldu \in U, \bldv \in V \}$. If $U \cap V = \{ \bldzero \}$, then
for any $\bldw \in U + V$ there is a unique representation in terms of the sum of two vectors $\bldw = \bldu + \bldv$, where $\bldu \in U$ and 
$\bldv \in V$. In this case we say that $U+V$ is a direct sum, and denote it by $U \oplus V$. 
It is easy to check that $\dim(U \oplus V) = \dim(U) + \dim(V)$.  

Let $W = U' \oplus U''$. For $V \subseteq W$ we define a projection of $V$ onto $U'$, denoted by $V|_{U'}$, as follows:
\[
V|_{U'} = \{ \bldu_1 \; : \; \bldu_1 + \bldu_2 \in V , \; \bldu_1 \in U', \; \bldu_2 \in U'' \} \; . 
\]
Similarly, we denote the projection of the vector $\bldu$ onto $U'$ by $(\bldu)|_{U'}$. 

For two vectors, $\bldu$ and $\bldv$, we write $\bldu \cdot \bldv$ to denote their scalar product. 
Let $W = U' \oplus U'' \subseteq \ff_q^n$ and assume that 
\begin{enumerate}
\item
$ U'' = \left< \blde_{i_1}, \blde_{i_2}, \cdots, \blde_{i_k} \right>$ for some $k$; 
\item
For all $\bldu \in U'$, $\bldv \in U''$, it holds $\bldu \cdot \bldv = 0$ (i.e. $U'$ and $U''$ are orthogonal).
\end{enumerate}
In that case, $U'$ is uniquely defined by $W$ and $U''$.  
For every vector $\bldu = (u_1, u_2, \cdots, u_n) \in U'$, we have $u_i = 0$ if $i \in \{ i_1, i_2, \cdots, i_k \}$. 
Then, we can define $\hat{U}$ as the subspace of $\ff_q^{n-k}$ obtained from $U'$ by removing all zero entries 
from the vectors in coordinates $\{ i_1, i_2, \cdots, i_k \}$.  
In that case we will write $W = \hat{U} \bigodot U''$. Observe that 
there is a natural bijection from the set of vectors in $U'$ onto the set of vectors in $\hat{U}$,
which is defined by removing all zeros in coordinates $\{ i_1, i_2, \cdots, i_k \}$.
In the sequel, sometimes we associate the vector in $\hat{U}$ with its 
pre-image in $U'$ under this bijection (or, simply speaking, sometimes we ignore the zero 
coordinates $\{ i_1, i_2, \cdots, i_k \}$ as above). Thus, by slightly abusing the notation we may also
write $W = U' \bigodot U''$.

Assume that $\dim(W) = n$. 
We use the notation $\cP(W,\ell)$ for the set of all subspaces of $W$ of dimension $\ell$, 
and $\cP(W)$ for the set of all subspaces of $W$ of any dimension. 
The number of $\ell$-dimensional subspaces of $W$, $0 \le \ell \le n$, 
is given by the $q$-ary Gaussian coefficient (see~\cite[Chapter 24]{Van-Lint}):
\[
|\cP(W,\ell)| = \left[ \begin{array}{c} n \\ \ell \end{array} \right]_q = \prod_{i=0}^{\ell-1} \frac{q^{n-i} - 1}{q^{\ell-i} - 1} \; . 
\]
For $U, V \in \cP(W)$, let
\[
\Distance(U, V) = \dim(U) + \dim(V) - 2 \dim(U \cap V)
\]
be a distance measure 
between $U$ and $V$ in the Grassmanian metric 
(see~\cite{KK}). We use the notation $\distance(\bldu, \bldv)$ for the Hamming distance between two vectors $\bldu$ and $\bldv$ of the same length. 

We say that the code $\Code$ is an $[n, \ell, \log_q(M), 2 D]_q$
subspace code, if it represents a set of subspaces in an ambient space $W$ over $\ff_q$, and satisfies the following conditions:
\begin{enumerate}
\item
$W$ is a vector space over $\ff_q$ with $\dim(W) = n$; 
\item
for all $V \in \Code$, $\dim(V) = \ell$;
\item
$|\Code| = M$;
\item
for all $U, V \in \Code$, $U \neq V$, it holds that
$\dim(U \cap V) \le \ell - D$, so that consequently $\Distance(U,V) \ge 2 D$.
\end{enumerate}
In~\cite{KK}, an $[\ell+m, \ell, mk, \ge 2 (\ell - k + 1)]_q$
subspace code was constructed by using an approach akin to Reed-Solomon codes. That code will henceforth be denoted by $\cK$. 
We refer the reader to~\cite{KK} for a detailed study of the code $\cK$. 

To formalize the network model, the authors of~\cite{KK} also introduced the  \emph{operator channel} and erasure operator as follows. 
Let $k \ge 0$ be an integer. Given a subspace $V \subseteq W$, if $\dim(V) \ge k$, the stochastic 
\emph{erasure operator} $\cH_k(V)$ returns some random $k$-dimensional subspace of $V$. Otherwise 
it returns $V$ itself. Then, for any subspace $U$ in $W$, it is always possible to write 
$U = \cH_k(V) \oplus E$, where $\cH_k(V)$ is a realization of $U \cap V$, $\dim(U \cap V) = k$,
and $E$ is a subspace of $W$. In particular, if $\dim(V) = k+1$, then $\cH_k(V)$ is called \emph{a dimension loss}. 
Similarly, if $\dim(V \oplus E) = \dim(V) + 1$ for some subspace $E$, then the corresponding operation
is called \emph{a dimension gain}. 

Decoding algorithms for the code $\cK$ were presented in~\cite{KK} and~\cite{SKK}. 
Suppose that $V \in \cK$ is transmitted over the operator channel. 
Suppose also that an $(\ell-\kappa+\gamma)$-dimensional subspace $U$ of $W$ is received, where
$k = \dim(U \cap V) = \ell - \kappa$. Here $\kappa = \ell - k$ denotes the number of dimension losses 
when modifying the subspace $V$ to $V \cap U$, while $\gamma$ similarly denotes the number of dimension gains needed to transform 
$V \cap U$ into $U$. Note that $\dim(E) = \gamma$, where $E$ is given in the decomposition of $U$. 
The decoders, presented in~\cite{KK} and~\cite{SKK}, are able to recover 
a single $V \in \cK$ whenever $\kappa+\gamma < D$. 
We denote hereafter a decoder for the code $\cK$ described in~\cite{KK} and~\cite{SKK} by $\cD_\cK$. Note that the decoding complexity of
$\cD_\cK$ is polynomial both in the dimension of the ambient vector space and the subspace distance $D$.

We find the following lemma useful in our subsequent derivations.

\begin{lemma}
Let $W = U' \oplus U''$ be a vector space over $\ff_q$, and let $V_1, V_2 \subseteq W$ be two vector subspaces. Then 
\[
\Distance(V_1, V_2) \ge \Distance(V_1|_{U'}, V_2|_{U'}) \; . 
\]
In other words, projections do not increase the subspace distance $\Distance$.
\label{lemma:projection-distance}
\end{lemma}

\begin{proof}
By definition, we have  
\begin{multline*}
\Distance(V_1, V_2) = (\dim(V_1) - \dim(V_1 \cap V_2)) \\ + (\dim(V_2) - \dim(V_1 \cap V_2)) \; .
\end{multline*}
Let $s = \dim(V_1 \cap V_2)$ and $t = \dim(V_1) - \dim(V_1 \cap V_2)$. Take $\{ \bldv_1, \bldv_2, \cdots, \bldv_s \}$ 
to be a basis of $V_1 \cap V_2$ and $\{ \bldu_1, \bldu_2, \cdots, \bldu_t \}$ to be $t$ linearly independent vectors 
in $V_1 \backslash V_2$. Then, $\{ \bldv_1, \bldv_2, \cdots, \bldv_s \}$ and $\{ \bldu_1, \bldu_2, \cdots, \bldu_t \}$ jointly 
constitute a basis of $V_1$. 

Next, consider 
\[
\cB_1 \define (V_1 \cap V_2)|_{U'} = \langle (\bldv_1)|_{U'}, (\bldv_2)|_{U'}, \cdots, (\bldv_s)|_{U'} \rangle \; .
\] 
Clearly, $(V_1 \cap V_2)|_{U'}$ is a subspace of $V_1|_{U'}$ and of $V_2|_{U'}$. Therefore, 
\[
(V_1 \cap V_2)|_{U'} \subseteq V_1|_{U'} \cap V_2|_{U'} .
\]
Note that the inclusion in the above relation may be strict. 
 
On the other hand, $\cB_1$ together with $\{ (\bldu_1)|_{U'}, (\bldu_2)|_{U'}, \cdots, (\bldu_t)|_{U'} \}$ 
spans $V_1|_{U'}$. We thus have that 
\begin{eqnarray}
&& \hspace{-15ex} \dim(V_1|_{U'}) - \dim(V_1|_{U'} \cap V_2|_{U'}) \nonumber
\\ & \le & \dim(V_1|_{U'}) - \dim(\cB_1) \nonumber \\ 
& \le & t \nonumber \\ 
& = & \dim(V_1) - \dim(V_1 \cap V_2) \; . 
\label{eq:dim-1}
\end{eqnarray}

Similarly to~(\ref{eq:dim-1}), it can be shown that 
\begin{equation}
\dim(V_2|_{U'}) - \dim(V_1|_{U'} \cap V_2|_{U'}) \le \dim(V_2) - \dim(V_1 \cap V_2) \; . 
\label{eq:dim-2}
\end{equation}

From~(\ref{eq:dim-1}) and~(\ref{eq:dim-2}), we obtain that 
\begin{eqnarray*}
&& \hspace{-5ex} \Distance(V_1|_{U'}, V_2|_{U'}) \\
& = & (\dim(V_1|_{U'}) - \dim(V_1|_{U'} \cap V_2|_{U'})) \\
& & \hspace{7ex} + \; (\dim(V_2|_{U'}) - \dim(V_1|_{U'} \cap V_2|_{U'})) \\
& \le & (\dim(V_1) - \dim(V_1 \cap V_2)) + (\dim(V_2) - \dim(V_1 \cap V_2)) \\
& = & \Distance(V_1, V_2) \; . 
\end{eqnarray*}
This completes the proof of the claimed result.
\end{proof}


\section{Hybrid Coding for Symbol Erasures and Dimension Gains/Losses}
\label{sec:code-L}

\subsection{Motivation}

Noncoherent network coding makes the topology of the network transparent to the code designer, 
and it has a strong theoretical foundations. Nevertheless, there are some practical issues that remain to be taken into account when
applying this coding scheme. First, the notion of ``dimension loss'' is fairly abstract since in networks only symbols (packets) can be 
erased or subjected to errors. It is reasonable to assume that a dimension loss corresponds to a number of 
symbol erasures/errors within the same message, although it is not clear how large this number is supposed to be. In the worst case scenario,
even one symbol erasure may lead to the change of one dimension. Second, the achievable throughput of the scheme and the underlying 
decoding complexity may be significantly inferior to those achievable only through classical network coding. Of course, this claim only holds
if the error-correcting scheme can be integrated with a linear network coding method. 

Let $W_\cL$ denote the space $\ff_q^n$ for some $n \in \nn$ 
and let $\cL$ be a set of subspaces of $W_\cL$ of dimension $\ell$. 
Assume that $V \in \cL$ is transmitted over a noncoherent network.
Assume that while propagating through the network, the vectors of $V$ were subjected to $\rho$ symbol errors and $\mu$ symbol erasures. 

Denote by $U$ the subspace spanned by the vectors obtained at the destination. 
Then, the vectors observed by the receiver are linear combinations of the vectors 
in $V$. Each of these vectors has, in the worst case scenario, at most $\rho$ symbol errors and $\mu$ symbol erasures.
Indeed, this can be justified as follows. 
If some vector $\bldx$ was transmitted in the network, and an erasure (or error) occurred in its $j$-th 
entry, in the worst case scenario this erasure (error) can effect only the $j$-th coordinates in \emph{all} vectors in 
$U$, causing this coordinate to be erased (or altered, respectively) in all of them. 
This is true for any network topology. 
Such erasure (or error) does not effect any other entries 
in the vectors observed by the receiver. 

We illustrate this concept by a simple example in Figure~\ref{fig:vectors}. 
In that example, one $\ff_q$-entry is erased in one vector.  
In the network-coding approach proposed in~\cite{KK}, in the worst-case scenario, 
this erasure is treated by disregarding the whole vector. In the proposed approach, 
in the worst-case scenario, the corresponding entry in all vectors will be erased. 

\begin{figure}[ht]
\begin{center}
  \subfloat[]{\includegraphics[width=0.4\textwidth]{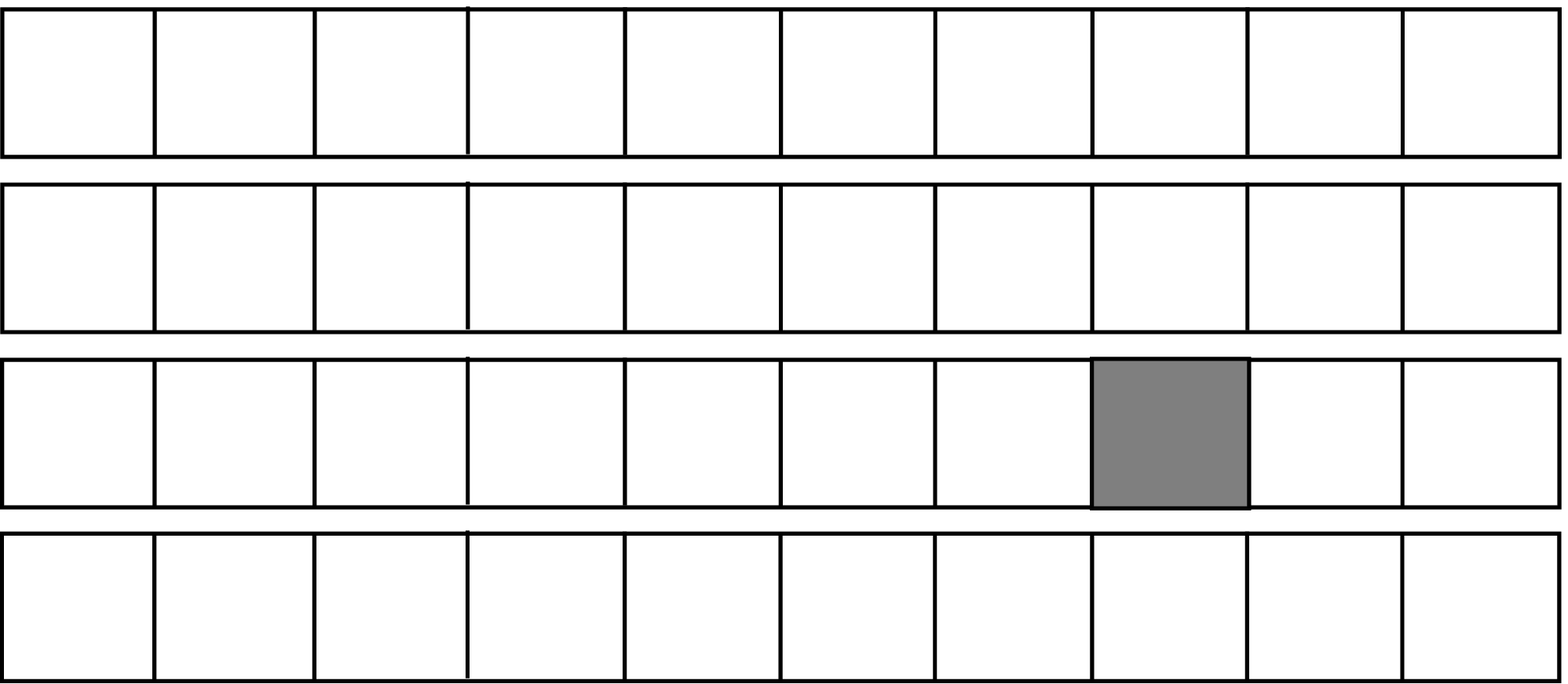}} 
  \vspace{3ex}
  \subfloat[]{\includegraphics[width=0.4\textwidth]{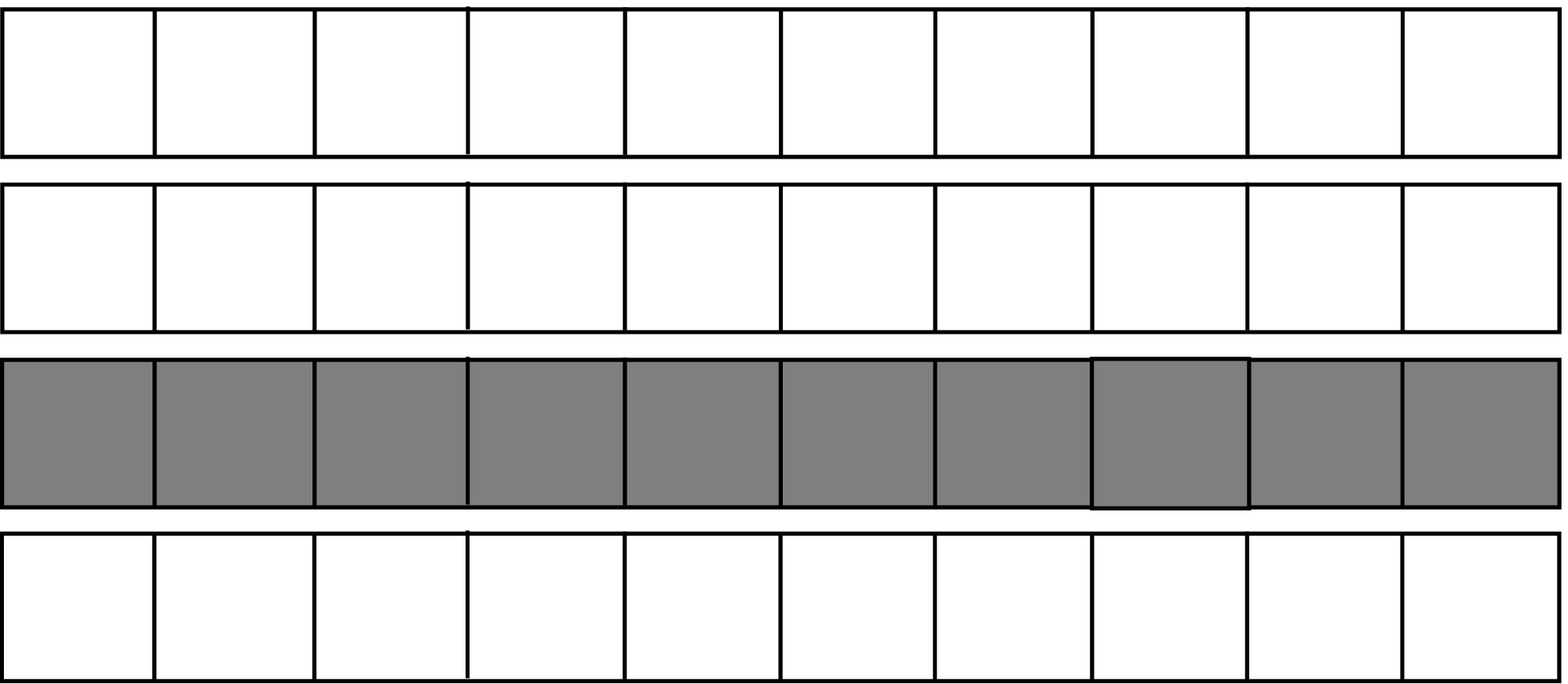}}
  \vspace{3ex}
  \subfloat[]{\includegraphics[width=0.4\textwidth]{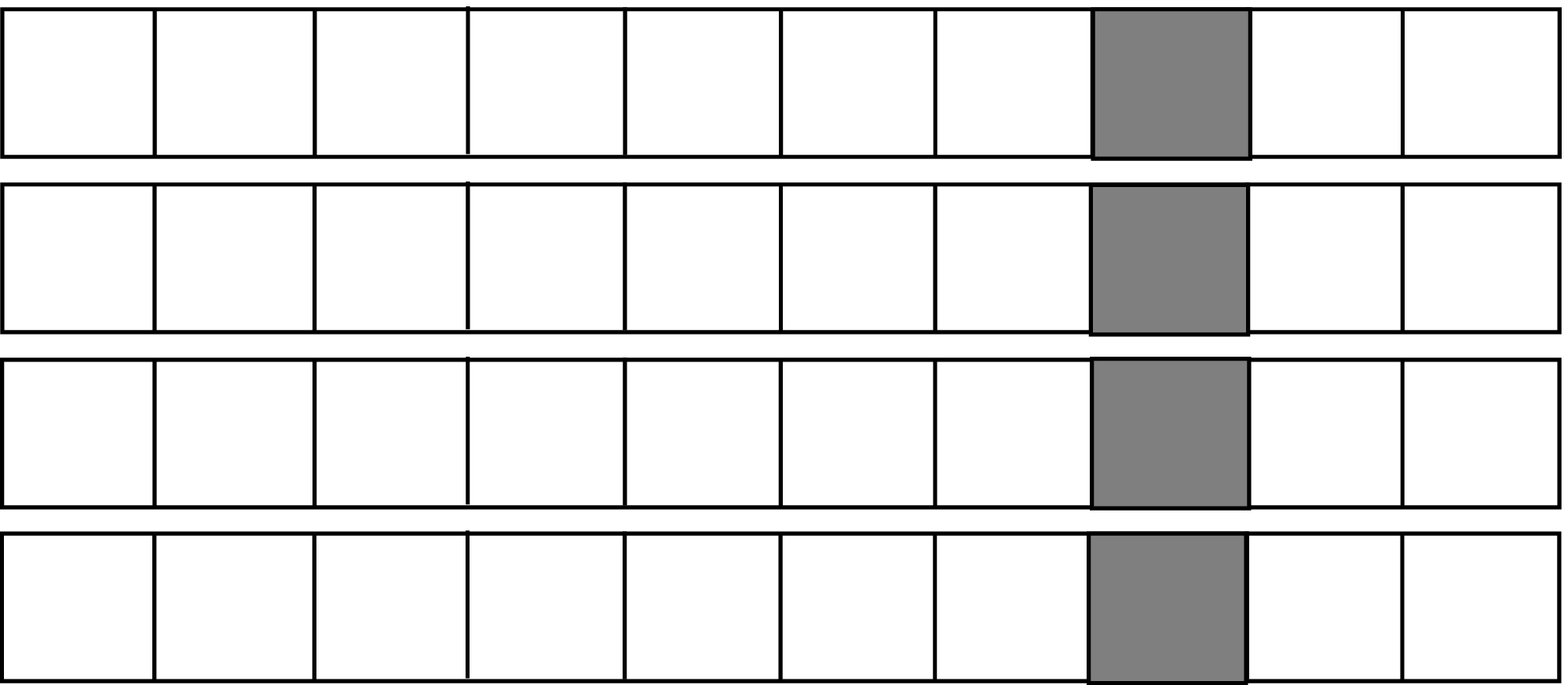}}
  \end{center}  
\caption{(a) Basis vectors of the transmitted subspace $V$: one vector has 
one entry erased. (b) In the approach of~\cite{KK}, this causes 
a dimension loss: the corresponding vector is discarded. (c) In the proposed 
alternative approach, the corresponding entry is erased in all vectors.}  
\label{fig:vectors}  
\end{figure}

This observation motivates the following definitions.

\begin{definition}
Consider a vector space $V \subseteq W_\cL$. Write $W_\cL = W_S \bigodot \langle \blde_j \rangle$ for some $1 \le j \le n$
and for some subspace $W_S$. 
A {\bf symbol error} in coordinate $j$ of $V$ is a mapping $\Phi_j$ from $V$ to $\Phi_j(V) = V' \subseteq W_\cL$, such that 
\[
V \neq V' \quad \mbox{ and } \quad V |_{W_S} = V' |_{W_S} \; . 
\]
\label{def:error}
\end{definition}
Observe that in general, one may have $\dim(V) \neq \dim(V')$ in Definition~\ref{def:error}. 

\medskip 

\begin{definition}
Let $V \subseteq W_\cL$ and assume that $W_\cL = W_S \bigodot \langle \blde_j \rangle$ for some $1 \le j \le n$
and for some subspace $W_S$. 
A {\bf symbol erasure} in coordinate $j$ of $V$ is a mapping $\Psi_j$ from $V$ to $\Psi_j(V) = V' \subseteq W_S$ 
such that 
\[
V|_{W_S} = V' \; . 
\]
\label{def:erasure}
\end{definition}

The subspace $V' \subseteq W_S$ can be naturally associated with $V_1 \subseteq (\ff_q \cup \{ ? \})^n$
in the following simple way: 
\begin{multline*}
    (v_1, \cdots, v_{j-1}, v_{j+1}, \cdots, v_n) \in V' \mbox{ if and only if } \\
    (v_1, \cdots, v_{j-1}, ?, v_{j+1}, \cdots, v_n) \in V_1 \;  . 
\end{multline*}
By slightly abusing the terminology, sometimes we say that $V_1$ (rather than $V'$) is obtained 
from $V$ by an erasure in coordinate $j$. Strictly speaking, such $V_1$ is not a vector space, 
since there are no mathematical operations defined for the symbol $?$. However, it 
will be natural to define for all $\bldu, \bldv, \bldw \in V_1$ and $\alpha \in \ff_q$ 
that $\bldu + \bldv = \bldw$ if and only if $(\bldu)|_{W_S} + (\bldv)|_{W_S} = (\bldw)|_{W_S}$, 
and that $\alpha \bldu = \bldw$ if and only if $\alpha \cdot (\bldu)|_{W_S} = (\bldw)|_{W_S}$. 

\medskip

To this end, we remark that there are four potential types of data errors in
a network that are not necessarily incurred independently: 
\begin{enumerate}
\item Symbol erasures;
\item Symbol errors; 
\item Dimension losses;
\item Dimension gains. 
\end{enumerate}
Below, we generalize the operator channel as follows.  
\begin{definition} 
Let $V$ be a subspace of $W = \ff_q^n$. The stochastic {\bf operator channel with symbol errors and erasures} 
returns a random subspace $U$ such that 
\begin{equation}
   U = \Upsilon_t ( \Upsilon_{t-1} ( \cdots ( \Upsilon_1 (V))) \subseteq \ff_q^{n - \mu} \; ,  
\label{eq:channel}
\end{equation}
for some $t$, where all $\Upsilon_i$, $i = 1, 2, \cdots, t$, are one of the following:
\begin{enumerate}
\item A symbol error $\Phi_j$, for some $j$; 
\item A symbol erasure $\Psi_j$, for some $j$, and there are exactly $\mu$ $\Upsilon_i$'s that are symbol erasures;
\item A dimension loss;
\item A dimension gain. 
\end{enumerate}
Furthermore, without loss of generality, we assume that all symbol errors and erasures are preceded by all dimension gains and losses\footnote{Even 
if the errors appear in a different order, the resulting subspace may still be generated by first applying dimension errors and then subsequently introducing symbol errors/erasures.}.
We reiterate this statement throughout the paper, since it is of importance in our subsequent derivations.
\label{def:channel}
\end{definition}

We now on focus on two important cases of Definition~\ref{def:channel}. 
\begin{definition} 
The operator channel in Definition~\ref{def:channel}
is called {\bf an operator channel with symbol erasures} if the number of symbol errors 
is always zero. 
\end{definition}
For an operator channel with symbol erasures, for each output $U$, we may and will always assume that 
dimension losses and gains have occurred first, followed by symbol erasures. 
More specifically, 
let $S = \{ j_1, j_2, \cdots, j_{\mu} \} \subseteq [n]$ be the set of erased coordinates in $V$, and let $W = W_S \bigodot \langle \blde_{j_1}, \blde_{j_2}, \cdots, \blde_{j_{\mu}} \rangle$. 
Then, 
\[
U_1 = \cH_k(V) \oplus E  \qquad \mbox{ and } \qquad U = U_1 |_{W_S} \; ,
\]
where $\dim(V \cap U_1) = k$, $\dim(U_1) = \ell'$. 
Here, one first creates $U_1$ from $V$ via dimension errors only (losses and gains). 
Subsequently, $U$ is created from $U_1$ by erasing coordinates in $S$. 
Let $\dim(V) = \ell$ and $\dim(U) = \ell'$. 
We say that $\ell - k$ is the number of dimension losses, $\ell' - k$ is the number of dimension gains,
and $\mu$ is the number of symbol erasures.   

\begin{definition} 
The operator channel in Definition~\ref{def:channel} is called {\bf an erasure-operator channel with symbol errors and erasures} if the number of dimension gains  
is always zero. 
\label{def:operator-channel-type-2}
\end{definition}
For the operator channel in Definition~\ref{def:operator-channel-type-2} we also assume that 
dimension losses have occurred first, followed by symbol errors and erasures. We define 
the numbers of dimension losses, symbol errors and symbol erasures analogous to the case 
of the operator channel with symbol erasures. 
\medskip 

In the forthcoming sections, we first concentrate on designing codes that are able to handle 
simultaneously symbol erasures, dimension losses and dimension gains. We postpone the discussion about how to handle 
symbol errors to Sections~\ref{sec:decoding-2} and~\ref{sec:decoding-failure}. 

\subsection{Code Definition}

We start the development of our approach with the following definition. 

\begin{definition}
A subspace code $\cL \subseteq \cP(W_\cL, \ell)$ (a set of subspaces in $W_\cL$ of dimension $\ell$) is called a {\bf code capable of correcting $D-1$ dimension errors} (either losses or gains) {\bf and $d-1$ symbols erasures},
or more succinctly, a $(D,d)$-{\bf hybrid code}, if it satisfies the following properties:   
\begin{enumerate}
\item
For any $V \in \cL$, $\dim(V) = \ell$.
\item
For any $U, V \in \cL$, $\dim(U) + \dim(V) - 2 \dim(U \cap V) \ge 2 D$. 
\item
Let $V \in \cL$. Let $V'$ be the subspace obtained from $V$ via $\mu$ symbol erasures, where 
$1 \le \mu \le d-1$. Then, $\dim(V') = \ell$ and 
the space $V$ is the only pre-image of $V'$ in $\cL$ under the given $\mu$ symbol erasures. 
\item
Let $U, V \in \cL$. Let $U', V'$ be obtained from $U$ and $V$, respectively, via $\mu$ symbol erasures, where 
$1 \le \mu \le d-1$. Here, both $U$ and $V$ have erasures in the same set of coordinates. 
Then, $\dim(U') + \dim(V') -2 \dim(U' \cap V') \ge 2 D$.
\end{enumerate}
\label{def:codes}
\end{definition}

We explain next why the class of  hybrid $(D,d)$ codes, satisfying properties 1) - 4), is capable of correcting $D-1$ dimension errors and $d-1$ symbol erasures. 
\begin{theorem} 
Let $\cL \subseteq \cP(W_\cL,\ell)$ be a code satisfying properties  1) - 4). 
Then, $\cL$ is capable of correcting any error pattern of $D-1$ dimension errors and $d-1$ symbol erasures. 
\end{theorem}

\begin{proof}
Suppose that $V \in \cL$ is transmitted through the operator channel, and that the subspace 
$U \in \cP(\ff_q^{n-d+1}, \ell')$ is received, where $D-1$ dimension errors and $d-1$ symbols erasures occurred. 
Note that here $\ell'$ is not necessarily equal to $\ell$. 

Recall the assumption that the dimension errors occurred first, and are followed by symbol erasures. 
As pointed out before, the order in which dimensional errors occur is irrelevant. 

More formally, let $S = \{ j_1, j_2, \cdots, j_{d-1} \} \subseteq [n]$ be  
a set of erased coordinates in $U$, and let $W_\cL = W_S \bigodot \langle \blde_{j_1}, \blde_{j_2}, \cdots, \blde_{j_{d-1}} \rangle$, for some subspace $W_S$. 
Then, 
\[
U_1 = \cH_k(V) \oplus E  \qquad \mbox{ and } \qquad U = U_1 |_{W_S} \; ,
\]
where $\dim(V \cap U_1) = k$, $\dim(U_1) = \ell'$, and 
\begin{equation}
\ell + \ell' - 2k \le D - 1 \; . 
\label{eq:dist-D-1}
\end{equation}

We show that if $\cL$ satisfies properties 1) - 4), then it is possible to recover $V$ from $U$. Indeed, consider the 
following set of subspaces 
\[
\cL' = \{ V|_{W_S} \subseteq \ff_q^{n-d+1} \; : \; V \in \cL \} \; .  
\] 
Take any $V_1, V_2 \in \cL'$. By property 3), $\dim(V_1) = \dim(V_2) = \ell$, and by property 4), 
$\dim(V_1) + \dim(V_2) - 2 \dim(V_1 \cap V_2) \ge 2D$. Therefore, $\cL'$ is a $[n-d+1, \ell, \log_q|\cL|, 2D]_q$ subspace code. 
It is capable of correcting of up to $D-1$ dimension errors in $\ff_q^{n-d+1}$. 

Denote $V' = V|_{W_S} \in \cL'$. Then, from Lemma~\ref{lemma:projection-distance} and the bound in~(\ref{eq:dist-D-1}),
\begin{eqnarray*}
\Distance(V',U) \le \Distance(V,U_1) = D - 1 \; .  
\end{eqnarray*}
We conclude that there exists a (not necessarily efficient) bounded-distance subspace decoder for the code $\cL'$ that is capable of recovering 
$V'$ from $U$. 

Finally, observe that $V'$ is obtained from $V$ by erasing $d-1$ coordinates indexed by $S$. 
From property~3), the pre-image of $V'$ under these erasures is unique. Therefore, $V$ can be recovered 
from $V'$. 
\end{proof}

\begin{remark}
The intuition behind the definition of hybrid codes is that dimension losses may and actually occur as a consequence of symbol erasures or errors.
Symbol erasures are ``easier'' to correct than dimension losses, and upon correcting a number of symbol erasures one expects to reduce the number of dimension losses. 
These claims are more rigorously formulated in Section 5.3.
\end{remark}

Throughout the remainder of the paper, we use the notation $[n, \ell, \log_q(M), 2D, d]_q$\footnote{Whenever it is apparent from the context, we omit the subscript $q$.} to denote a 
hybrid code $\cL \subseteq \cP(W, \ell)$ with the following properties: 
\begin{enumerate}
\item
$\dim(W) = n$; 
\item
for all $V \in \code$, $\dim(V) = \ell$;
\item
$|\cL| = M$;
\item
$\cL$ is a code capable of correcting $D-1$ dimension errors and $d-1$ symbols erasures. 
\end{enumerate}

\section{Bounds on the Parameters of Hybrid Codes}
\label{sec:bounds}

In this section, we derive the Singleton and the sphere-packing bound for hybrid codes handling dimension losses/gains and symbol erasures simultaneously.

\subsection{The Singleton Bound}

Assume that a vector space $W$ over $\ff_q$ has dimension $n$, and let $\cL \subseteq \cP(W, \ell)$
be a subspace code. In what follows, we use a \emph{puncturing} of the code $\cL$.

\begin{definition} The puncturing of a code $\cL \subseteq \cP(W, \ell)$ at position $j$ is 
a set $\cL'$ of subspaces of $\ff_q^{n-1}$ given by 
\begin{equation}
\cL' \triangleq \{ \Psi_j(V) \; : \;  V \in \cL \} \; .  
\label{eq:puncturing}
\end{equation}
\label{def:puncturing}
\end{definition}

\begin{remark}
In general, Definition~\ref{def:puncturing} is different from the definition of puncturing in~\cite[Section IV.C]{KK}.  
In particular, the puncturing in Definition~\ref{def:puncturing} does not necessarily decrease the dimension of $V$. On the other hand, Definition~\ref{def:puncturing} is similar to the first part of the definition of \emph{$j$-coordinate puncturing} in~\cite[Section 5.A]{Etzion-Silberstein}. 
\end{remark} 

\begin{theorem}
Let $\cL$ be a code of type $[n, \ell, \log_q (M), 2D, d]$ in the ambient space $W_\cL$. 
If $d > 1$, then coordinate puncturing at coordinate $j$  
yields a code with parameters $[n-1, \ell, \log_q (M), 2D, \ge d -1 ]$.
\end{theorem}

\begin{proof}
Let $\cL'$ be a code obtained by puncturing of the $j$-th coordinate in all vectors spaces in $\cL$, 
as in~(\ref{eq:puncturing}). 
Clearly, the dimension of the ambient space decreases by one under this puncturing, 
and so the resulting ambient space $W'$ satisfies $\dim(W') = n-1$.  

Let $V \in \cL$. Since $d>1$, by property (3) in Definition~\ref{def:codes}, $\dim(V') = \ell$ and 
$V'$ has a unique pre-image. Therefore, $|\cL| = |\cL'|$. 

The fact that puncturing does not change the subspace distance $2D$  follows from the property that $\cL$ is a code capable of correcting 
$D-1$ dimension errors and $d-1$ 
symbol errors. Thus, $\dim(U') + \dim(V') -2 \dim(U' \cap V') \ge 2 D$,
where $U'$ and $V'$ are obtained by puncturing of $U$ and $V$, respectively.  
Since each subspace in $\cL'$ is obtained from its pre-image in $\cL$ by an erasure in the $j$-th coordinate, the codes' Hamming distance resulting
from puncturing is at least $d-1$. 
\end{proof}

\begin{theorem}
The size $M$ of the $[n, \ell, \log_q (M), 2D, d]_q$ code $\cL$ satisfies
\[
M \le \cA_q(n - d + 1, \ell, 2D) \; , 
\]
where $\cA_q(n, \ell, 2D)$ stands for the size of the largest subspace code $[n, \ell, M', 2D]_q$.  
\end{theorem}

\begin{proof}
We apply $d-1$ coordinate puncturings to $\cL$. 
The resulting code is a $[n-d+1, \ell, \log_q(M), 2D]$ subspace code. 
Indeed, it has the same number of codewords as $\cL$, and it is a set of 
$\ell$ dimensional subspaces in a $n-d+1$-dimensional ambient space, whose 
pairwise intersection is of dimension $\le \ell-D$. In particular, its size is upper bounded 
by $\cA_q(n - d + 1, \ell, 2D)$. 
\end{proof}

\begin{corollary}
From the Singleton bound in~\cite{KK}, the size $M$ of the $[n, \ell, \log_q (M), 2 D, d]_q$ code $\cL$ satisfies
\begin{equation}
M \le \min \left\{ \left[ \begin{array}{c} n - d - D + 2 \\ \ell - D + 1 \end{array} \right]_q, 
\left[ \begin{array}{c} n - d - D + 2 \\ \ell \end{array} \right]_q \right\}\; . 
\label{eq:singleton}
\end{equation}
\label{cor:singleton}
\end{corollary}

We use the following result from~\cite{KK}. 
\begin{lemma}[Lemma 4 in~\cite{KK}]
The Gaussian coefficient ${n \brack \ell}_q$ satisfies
\[
1 <  q^{-\ell(n - \ell)} {n \brack \ell}_q < 4 \; . 
\]
\label{lemma:gauss-coeff}
\end{lemma}

We also use the following definition of the rate of the subspace code. 
\begin{definition}
The {\bf rate} of the subspace code $\cL$ is defined as $R = \frac{\log_q (|\cL|)}{n \ell}$. 
\end{definition}

Next, let 
\[
\lambda = \frac{\ell}{n}, \; \Delta = \frac{D}{\ell} \mbox{ and } 
\delta = \frac{d}{n} \; . 
\]
Thus, by using the bound in~(\ref{eq:singleton}), an asymptotic version of the latter bound reads as follows. 
\begin{corollary}
The rate of a $[n, \ell, \log_q(|\cL|), 2D, d]_q$ code $\cL$ satisfies
\begin{equation*}
\label{eq-rate}
R \le \left( 1 - \Delta + \frac{1}{\ell} \right) \left(1 - \delta - \lambda + \frac{1}{n} \right) + o(1)\; . 
\end{equation*}
\end{corollary}

\begin{proof}
We start with the first expression on the right-hand side of~(\ref{eq:singleton}), 
namely
\[
|\cL| \le \left[ \begin{array}{c} n - d - D + 2 \\ \ell - D + 1 \end{array} \right] \; .
\]
From Lemma~\ref{lemma:gauss-coeff}, we obtain that
\[
|\cL| < 4 \cdot q^{ (\ell - D + 1)(n - d - \ell + 1)} \; .
\]
Taking $\log_q(\cdot)$ of both sides yields
\[
\log_q(|\cL|) < \log_q(4) + (\ell - D + 1)(n - d - \ell + 1) \; ,
\]
and the required result is obtained by dividing the last inequality by $n \ell$, i.e.
\begin{eqnarray*}
R = \frac{\log_q(|\cL|)}{n \ell} & < & \frac{\ell - D + 1}{\ell} \cdot \frac{n - d - \ell + 1}{n} + o(1) \\
& = & \left( 1 - \Delta + \frac{1}{\ell} \right) \left(1 - \delta - \lambda + \frac{1}{n} \right) + o(1)
\; .
\end{eqnarray*}

\end{proof}

\subsection{Sphere-Packing Bound}

We start with the following definition. 
\begin{definition} 
A matrix $\bldM$ over $\ff_q$ is said to be in a {\bf reduced row echelon form} if the following conditions hold:
\begin{itemize}
\item
Each nonzero row in $\bldM$ has more leading zeros than the previous row.
\item
The leftmost nonzero entry in each row in $\bldM$ is one.
\item
Every leftmost nonzero entry in each row is the only nonzero entry in its
column.
\end{itemize}
\end{definition}
It is well known that any $\ell$-dimensional subspace of $\ff_q^n$ can be \emph{uniquely} represented by a 
$\ell \times n$ matrix over $\ff_q$ in reduced row echelon form.
\medskip

Let $W_\cL$ be the ambient space $\ff^n_q$, and let $0 \le \ell \le n$. 
Fix two integers $T \in [0, 2\ell]$, and $t \in [0, n-1]$. Two vector spaces $U, V \in \cP(W_\cL, \ell)$ are called 
\emph{$(T, t)$-adjacent} if there exists a set of coordinates 
$S = \{ i_1, i_2, \cdots, i_s\} \subseteq [n]$, $s \le t$, and a vector space $W_{S}$ such that 
\[
 W_\cL = W_{S} \bigodot \langle \blde_{i_1}, \blde_{i_2}, \cdots, \blde_{i_s} \rangle \; ,  
\] 
and
\[
\Distance(U|_{W_S}, V|_{W_S}) \le T \; . 
\]
Note that the adjacency relation is symmetric with respect to the order of $U, V$, namely $U$ and $V$ are 
$(T, t)$-adjacent if and only if $V$ and $U$ are $(T, t)$-adjacent. 
 
Assume that the $[n, \ell, \log_q(|\cL|), 2D, d]_q$ 
hybrid code $\cL$ is used over a network. Let $V \in \cL$ be the transmitted subspace, 
and let $U|_{W_S} \subseteq \ff_q^{n-t}$ be the received subspace, 
for some $U \in \cP(W_\cL, \ell)$, and for some $W_S$ as above, as a result of $T$ dimension losses and gains, 
and $t$ symbol erasures. Then, $U$ and $V$ are $(T, t)$-adjacent. If there is no other 
codeword $\tilde{V} \in \cL$ such that $\tilde{V}$ and $U$ are $(T, t)$-adjacent, then the decoder, which
is able to correct $T$ dimension erasures/gains and $t$ symbol erasures, can recover $V$ 
from $U$. This observation motivates the following definition. 

\begin{definition}
Let $W_\cL$ be a vector space $\ff_q^n$, and let $V \in \cP(W_\cL, \ell)$. The sphere $\cS(V, \ell, T, t)$ around $V$ is   
defined as 
\begin{multline*}
\cS(V, \ell, T, t) = \{ U \in \cP(W_\cL, \ell) \; : \; \\ V \mbox{ and } U \mbox{ are  $(T, t)$-adjacent} \} \; . 
\end{multline*}
\end{definition}

Now, we recall the following result from~\cite{KK}.  
\begin{theorem}[Theorem 5 in~\cite{KK}]
For any $V \in \cP(W_\cL, \ell)$, and any $0 \le T \le 2 \ell$, 
\[
|\cS(V, \ell, T, 0)| = \sum_{i=0}^{T/2} q^{i^2} {\ell \brack i} {n-\ell \brack i} \; . 
\]
\end{theorem}
We generalize this theorem in the following way. 
\begin{theorem}
\label{thrm:sphere-size}
Let $\cL$ be a $[n, \ell, \log_q(|\cL|), 2D, d]_q$ code. 
For any $V \in \cL$, any $0 \le T \le 2 \ell$, and any $0 \le t \le d-1$, 
\[
|\cS(V, \ell, T, t)| \ge q^{\ell t} \cdot \sum_{i=0}^{T/2} q^{i^2} {\ell \brack i} {n-t-\ell \brack i} \; .
\]
\end{theorem}
\begin{proof}
Take a set $S = \{ n-t+1, n-t+2, \cdots, n \} $ of cardinality $t$. Let $W_\cL = \ff^n_q$, 
and take a vector space $W_{S}$ given by 
\[
 W_\cL = W_{S} \bigodot \langle \blde_{n-t+1}, \blde_{n-t+2}, \cdots, \blde_{n} \rangle \; .  
\] 
Fix some $V \in W_\cL$ and consider an arbitrary $U \in W_\cL$, such that 
$V$ and $U$ are $(T,t)$-adjacent. Define $V' \triangleq V|_{W_S} \subseteq \ff_q^{n-t}$ and 
$U' \triangleq U|_{W_S} \subseteq \ff_q^{n-t}$. 
Then, by the definition of $(T,t)$-adjacency, 
\begin{equation}
\Distance(V', U') \le T \; . 
\label{eq:dist-T}
\end{equation}
Therefore, for a given $V'$, the number of subspaces $U'$ satisfying~(\ref{eq:dist-T}) is given by  
\[
\sum_{i=0}^{T/2} q^{i^2} {\ell \brack i} {n-t-\ell \brack i} \; .
\]

Next, we estimate the number of different subspaces $U \in \ff_q^n$ such that $U|_{W_S} = U'$ for a given $U'$. 
Consider the reduced row echelon form matrix $\bldM'$ of dimension $\ell \times (n-t)$ representing $U'$. In order to obtain 
this matrix, we need to establish the values 
of the last $t$ entries in every row of $\bldM$, while the first $n-t$ entries are equal to their counterparts
in $\bldM'$. There are $\ell$ rows in $\bldM$, and each entry can take one of $q$ values. 
As each $U'$ yields $q^{\ell t}$ different choices of $U$, the claimed result follows. 
\end{proof}

The following example further illustrates the idea of the proof. 
\begin{example}
Assume that $n = 9$, $\ell = 3$, $T=2$ and $t=4$, and that the reduced row echelon form of $V$ is given by 
the following matrix over $\ff_2$: 
\begin{eqnarray*}
\left( \begin{array}{ccccc|cccc} 
1 & 0 & 0 & 1 & 0 & 0 & 1 & 1 & 0 \\
0 & 1 & 1 & 0 & 0 & 1 & 1 & 1 & 0 \\
0 & 0 & 0 & 0 & 1 & 0 & 0 & 1 & 0 
\end{array}\right) \; . 
\end{eqnarray*}
Then, the reduced row echelon form of $V'$ is given by the first five columns of the above matrix. 
Consider $U'$ such that $\Distance(V',U') \le 2$. Such a $U'$ can be obtained when one 
of the rows in the reduced row echelon form of $V'$ is replaced by a different vector.
One possible reduced row echelon form of $U'$ is the sub-matrix formed by the first 5 columns of 
the matrix below.  
Then, in order to find all possible options for $U$, we need to fill in the values of the black dots 
in the last $t=4$ columns of the following matrix
\begin{eqnarray*}
& \left. \left( \begin{array}{ccccc|cccc} 
1 & 0 & 1 & 1 & 0 & \bullet & \bullet & \bullet & \bullet \\
0 & 1 & 1 & 0 & 0 & \bullet & \bullet & \bullet & \bullet \\
0 & 0 & 0 & 0 & 1 & \bullet & \bullet & \bullet & \bullet 
\end{array}\right)   
\right\} \ell . \\ 
& \underbrace{\phantom{ooooooooooooooo}}_{n-t} \underbrace{\phantom{oooooooooooo}}_{t} \phantom{oo}
\end{eqnarray*}
This can be done in $q^{12}$ ways. 

\end{example}

From Theorem~\ref{thrm:sphere-size}, the following sphere-packing bound is immediate. 
\begin{corollary}
Let $\cL \subseteq \cP(W_\cL, \ell)$ be a code that corrects $D-1$ dimension losses/gains and $d-1$ symbol erasures. 
Then, for all $V \in \cL$, the spheres $\cS(V,\ell, D-1, d-1)$ are disjoint. Therefore,  
\begin{eqnarray}
|\cL| & \le & \frac{|\cP(W_\cL, \ell)|}{|\cS(V,\ell, D-1, d-1)|} \nonumber \\ 
& \le &  
\frac{{n \brack \ell}}{  q^{\ell (d-1)} \cdot \sum_{i=0}^{(D-1)/2} q^{i^2} {\ell \brack i} {n-(d-1)-\ell \brack i}} \; . 
\label{eq:sphere-packing}
\end{eqnarray}
\end{corollary}

Now, we turn to an asymptotic analysis of the bound~(\ref{eq:sphere-packing}). From Lemma~\ref{lemma:gauss-coeff}, we have 
\begin{eqnarray*}
|\cL| & \le & 
\frac{4 q^{\ell(n- \ell)}}{q^{\ell (d-1)} \cdot \sum_{i=0}^{(D-1)/2} q^{i^2 + i(\ell-i) + i(n-(d-1)-\ell-i)}} \\ 
& = & \frac{4 q^{\ell(n- \ell)}}{\sum_{i=0}^{(D-1)/2} q^{\ell (d-1)} \cdot q^{i(n-(d-1)-i)}} 
\; . 
\end{eqnarray*}
If $D-1 \le n - (d-1)$, then the dominant term in 
$\sum_{i=0}^{(D-1)/2} q^{i(n-(d-1)-i)}$ is obtained when $i = (D-1)/2$.
In that case, one has
\begin{eqnarray*}
|\cL| & \le & \frac{4 q^{\ell(n- \ell)}}{q^{\ell (d-1)} \cdot q^{(D-1)(n-(d-1)-(D-1)/2)/2}} \\
& \equival & 4 q^{ \ell(n - d - \ell + 1) - (D-1)(n-d-(D-1)/2+1)/2}
\; , 
\end{eqnarray*}
where $f(x) \equival g(x)$
means that the two expressions $f(x)$ and $g(x)$ are asymptotically equal.  

By taking the base-$q$ logarithm and dividing both sides of the above expression by $\ell n$, we obtain the following result. 
\begin{corollary}
Let $\cL \subseteq \cP(W_\cL, \ell)$ be a code that corrects $D-1$ dimension losses/gains and $d-1$ symbol erasures. 
Then, its rate satisfies: 
\begin{multline*}
R \le \left(1 - \delta - \lambda + \frac{1}{n} \right) \\ 
- \left(\frac{\Delta}{2} - \frac{1}{2 \ell} \right) \left(1 - \delta - \frac{\lambda \Delta}{2} + \frac{3}{2n} \right) 
\; + \; o(1) \; . 
\end{multline*}
\end{corollary}

\section{Code Construction}
\label{sec:construction}

Next, we construct hybrid codes capable of correcting dimension losses/gains and symbol erasures simultaneously. 
We show that these code are asymptotically optimal with respect 
to the Singleton bound. We also provide some examples comparing hybrid codes to subspace codes. 

\subsection{Code Construction}
\label{sec:construct}

Let $W$ be a vector space $\ff_q^{\ell + m}$ of dimension $\ell + m$, and let $\Code$ be 
a $[\ell + m, \ell, \log(|\Code|), 2D]$ subspace code in $W$. In other words, 
$\Code$ is a set of subspaces of $W$ of dimension $\ell$, such that for any $U, V \in \Code$, $V \neq U$, 
$\dim(U \cap V) \le \ell - D$. We fix a basis of $W$, and denote its vectors 
by $\{ \bldu_1, \bldu_2, \cdots, \bldu_{\ell + m} \}$. We denote the decoder for the subspace metric 
and $\Code$ by $\cD_\Code$. 

Let $\bldG$ be a $(\ell+m) \times n$ generator matrix of the $[n, \ell+m, d]$ Generalized Reed-Solomon (GRS) code $\code$ over $\ff_q$ of length $n \define \ell+m+d-1$, given by 
\begin{multline*}
\bldG = \left( \begin{array}{ccccc}
1 & 1 & 1 & \cdots & 1 \\
a_1 & a_2 & a_3 & \cdots & a_n \\
a_1^2 & a_2^2 & a_3^2 & \cdots & a_n^2 \\
\vdots & \vdots & \vdots & \ddots & \vdots \\
a_1^{\ell+m-1} & a_2^{\ell+m-1} & a_3^{\ell+m-1} & \cdots & a_n^{\ell+m-1} \\ 
\end{array} \right) \\ 
\cdot 
\left( \begin{array}{cccc}
\eta_1 &  &  & 0 \\
& \eta_2 &  & \\
& & \ddots & \\
0 & & & \eta_n \\ 
\end{array} \right). 
\end{multline*}
Here, $a_i \in \ff_q$, $1 \le i \le n$, denote $n$ distinct nonzero field elements, while $\eta_i \in \ff_q$, $1 \le i \le n$, denote arbitrary nonzero elements (see~\cite[Chapter 5]{Roth-book} for more details).  

We use the notation $\bldG_i$ for the $i$-th row of $\bldG$, for $i = 1, 2, \cdots, \ell+m$. 
The code $\code$ is capable of correcting any 
error pattern of $\rho$ errors and $\mu$ erasures given that $2 \rho + \mu \le d-1$. In this section, we are particularly interested in the case when $\rho=0$.

Denote by $\cD_{RS}$ a decoder for the code $\code$, which corrects any 
error pattern of $\rho$ errors and $\mu$ erasures, whenever $2 \rho + \mu \le d-1$. 
Denote by $W_\cL$ the linear space $\ff_q^n$. 
Let $\bldA$ be a $(\ell+m) \times (\ell+m)$ matrix over $\ff_q$ such that 
\[
\forall i = 1, 2, \cdots, \ell+m \; : \; \blde_i = \bldu_i \bldA \; , 
\]
and therefore 
\[
\forall i = 1, 2, \cdots, \ell+m \; : \; \bldG_i = \bldu_i \bldA \bldG \; . 
\]
Clearly, such an $\bldA$ exists since $\{ \blde_i \}$ and $\{ \bldu_i \}$ are two different bases for $\ff_q^{\ell+m}$. 

We define a linear mapping $\cE_\cL : W \rightarrow W_\cL$ as follows. For an arbitrary vector $\bldv \in W$, 
\[
\cE_\cL (\bldv) = \bldv \bldA \bldG \; . 
\]
This mapping can be naturally extended to the mapping 
$\cE_\cL : \cP(W) \rightarrow \cP(\code)$ (with the slight abuse of notation), 
where $\cP(\code)$ stands for 
a set of all linear subcodes of $\code$. For any $V \in \cP(W)$, we have 
\[
\cE_\cL (V) \define \{ \bldv \bldA \bldG \; : \; \bldv \in V \} \in \cP(\code) \; . 
\]
It is easy to see that $\cE_\cL$ is a linear mapping, and that the image of the linear space $V$ is a linear space.  
Moreover, it is straightforward to show that this mapping, when applied to subspaces of $W$, is one-to-one. 
Thus, for any $V \in W$, 
\begin{equation}
\dim(V) = \dim( \cE_\cL(V) ) \; . 
\label{eq:mapping-1}
\end{equation}
One can check that for any $U,V \in W$, it holds
\begin{equation}
\dim(U \cap V) = \dim( \cE_\cL(U) \cap \cE_\cL(V) ) \; . 
\label{eq:mapping-2}
\end{equation}

Next, we define a code $\cL \subseteq \cP(W_\cL, \ell)$ as 
\[
\cL = \left\{ \cE_\cL(V)  \; : \; V \in \Code \right\} \; .   
\]

\begin{theorem}
The code $\cL$ is a hybrid code over $\ff_q$, with parameters $[n, \ell, \log_q(| \Code |), \ge 2D, \ge d]$. 
\end{theorem}

\begin{proof}
It is straightforward to verify that $\cL$ has the first two parameters stated. 
The third parameter follows from the fact that $|\Code|$ is the number of subspaces
in $\Code$, and two different subspaces are mapped onto different subspaces 
under $\cE_\cL$. 

Next, we show that $\cL$ is a code capable of  correcting $D-1$ dimension gains/losses and $d-1$ symbols erasures.
It suffices to show the following two properties: 

\begin{enumerate}
\item
Let $V \in \cL$. Let $V'$ be the subspace obtained from $V$ by any $\mu$ symbol erasures, such that 
$\mu \le d-1$. Then, $\dim(V') = \ell$ and 
the space $V$ is the only pre-image of $V'$ in $\cL$. 
\item
Let $U, V \in \cL$. Let $U', V'$ be obtained from $U$ and $V$, respectively,
by $\mu$ symbol erasures, such that $\mu \le d-1$. 
Then, $\dim(U') + \dim(V') -2 \dim(U' \cap V') \ge 2 D$.
\end{enumerate}
Indeed, that these two conditions are satisfied can be shown as follows. 
\begin{enumerate}
\item
Let $V \in \cL$ and $V'$ be obtained from $V$ by $\mu$ symbol erasures, such that 
$\mu \le d-1$. Let $\{\bldv_1, \bldv_2, \cdots, \bldv_\ell\}$ be a basis of $V$, and let 
$\{\bldv'_1, \bldv'_2, \cdots, \bldv'_\ell\}$ be a set of corresponding vectors obtained by 
$\mu$ symbol erasures.  
Then, for any $a_1, a_2, \cdots, a_\ell \in \ff_q$, not all of which are zero,  
\[
\bldv = \sum_{i = 1}^\ell a_i \bldv_i \in V  
\]
is a vector of the Hamming weight $\ge d$. Therefore, after applying $\mu$ symbol erasures, 
the Hamming weight of the resulting vector
\[
\bldv' = \sum_{i = 1}^\ell a_i \bldv'_i  
\]
is at least $d-(d-1) \ge 1$, for any $a_1, a_2, \cdots, a_\ell \in \ff_q$, not all of which are zero. 
Therefore, the vectors $\{\bldv'_1, \bldv'_2, \cdots, \bldv'_\ell\}$ are linearly independent, and thus 
$\dim(V') = \ell$. 

Next, take a vector $\bldv' = \sum_{i = 1}^\ell a_i \bldv'_i \in V'$. Since the minimum distance of a code $\code$ is $d$, 
and thus the code can correct any pattern of up to $d-1$ symbol erasures, 
the only possible pre-image of $\bldv'$ under any $\mu$ symbol erasures, $\mu \le d-1$, is 
$\bldv = \sum_{i = 1}^\ell a_i \bldv_i \in V$. 
Therefore, each $V' \in \cL'$ has a unique pre-image. 

\item
Let $U, V \in \cL$ and let $U', V'$ be obtained from $U$ and $V$, respectively,
by $\mu$ symbol erasures, such that 
$\mu \le d-1$.

From part (1) of the proof, $\dim(U') = \dim(V') = \ell$. It is sufficient to show that 
$\dim(U' \cap V') \le \dim(U \cap V)$. Assume, on the contrary, that $\dim(U' \cap V') > \dim(U \cap V)$. 
This means that there exists $\bldu \in U$ and $\bldv \in V$, $\bldu \neq \bldv$, such that 
by applying $\mu$ symbol erasures to these vectors, one obtains resulting vectors $\bldu'$ and $\bldv'$ that are equal. 
Recall, however, that $\bldu, \bldv \in \code$, and therefore $\distance(\bldu, \bldv) \ge d$. 
We hence arrive to a contradiction, and thus $\dim(U') + \dim(V') -2 \dim(U' \cap V') \ge 2 D$.
\end{enumerate}
\end{proof}

The following generalization of Property~1) above holds. 

\begin{corollary}
Let $V$ be a subcode of $\code$ (of any dimension). Let $V'$ be obtained from $V$ by arbitrary $\mu$ symbol erasures, such that 
$\mu \le d-1$. Then, $\dim(V') = \dim(V)$ and 
the space $V$ is the only pre-image of $V'$ in $\cP(\code)$. 
\label{cor:preimage}
\end{corollary}
The proof follows along the same lines of the proof of Property~1). 

\subsection{Asymptotic Optimality}

In Section~\ref{sec:bounds} we derived upper bounds on the size of general hybrid codes. 
These bounds imply upper bounds on the size of the code $\cL$. Moreover, the code $\cL$ 
is asymptotically optimal with respect to one of these bounds, as will be shown below. 

Consider the code $\cL$ constructed from a subspace code with 
parameters $[\ell + m, \ell, \log_q|\Code|, 2D]_q$ and a classical GRS code
with parameters $[n, \ell + m, d]_q$, $n = \ell + m + d -1$, as described in the previous section. 
The resulting code $\cL$
is a $[n, \ell + m, \log_q|\Code|, 2D, d]_q$ code. 
The number of codewords of the code is $|\Code|$. If we take $\Code$ as described in~\cite{KK}, with $\ell \le m$, then
the $q$-ary logarithm of the number of the codewords is given by 
\[
\log_q |\Code| = m (\ell - D + 1) \; .
\]
Therefore, the log of the cardinality of $\cL$ is given by 
\begin{equation}
\log_q |\cL| = m (\ell - D + 1) = (n - \ell - d + 1) (\ell - D + 1) \; . 
\label{eq:number-of-words}
\end{equation}
Hence from Lemma~\ref{lemma:gauss-coeff} we have the following result. 
\begin{corollary}
\[
\frac{1}{4} {n - d - D + 2 \brack \ell - D + 1 }_q < |\cL| < { n - d - D + 2 \brack \ell - D + 1 }_q \;. 
\]
Thus, the code $\cL$ is asymptotically order-optimal (i.e., optimal up to a constant factor) with respect to the Singleton bound~(\ref{eq:singleton}).
\label{cor:singleton-optimal} 
\end{corollary}

\subsection{Hybrid Codes versus Subspace Codes}

In this section, we analyze the error-correcting ability of hybrid codes and subspace codes. We show that 
in some scenarios, hybrid codes achieve significant improvement in code rate, while having 
similar error-correcting capability, when compared to their subspace codes counterparts. 

\subsection*{Examples} 

Let $\cK$ be the code defined as in~\cite{KK}. 
When the code $\cK$ is used over a noncoherent network, in the worst case scenario each symbol erasure translates into 
the loss of one dimension, and each symbol error translates into one dimension loss and one erroneous dimension gain. This may 
happen when all the erasures and errors occur in linearly independent vectors. In addition, note that requiring each linearly independent 
vector to be able to correct up to and including $d-1$ erasures is somewhat restrictive, since it imposes an individual, rather than joint constraint,
on the total number of erasures in the transmitted subspace.

We show next the advantage of using the code $\cL$ for the case when all data errors in the noncoherent network take form of symbol erasures. These
symbol erasures are the cause of dimension losses. The code $\cL$ has more codewords than $\cK$ while having the same overall 
error-correcting capability. 

\begin{example}
Consider the code $\cK$ with parameters $[\ell + m, \ell, mk, 2 (\ell - k + 1)] = [6, 3, 3, 6]_q$. This code can correct up to and including two dimension losses. 
If the symbol erasures happen in the linearly independent vectors, the result is a loss of 
two dimensions, and the code can provably correct such two symbol erasures. Alternatively, 
one symbol error results in one dimension loss and one dimension gain, which is also correctable by this code.
However, this code is not able to correct any combination of 
three symbol erasures that occur in different basis vectors. Note that the code contains $q^3$ (subspaces) codewords. 

Now, let $W = \ff_q^4$ and consider the set $\cP = \cP(W, 3)$, where $|\cP| = {4 \brack 3}_q$. 
Fix some basis for $\cP$. Let $\code$ be $[6, 4, 3]_q$ GRS code (for $q \ge 5$). 
Define the mapping $\cE_\cL : W \rightarrow \code$ as in Section~\ref{sec:construct}. 

The resulting code $\cL$ has 
\[
{4 \brack 3}_q = \frac{q^4 - 1}{q-1} 
\]
codewords (subspaces), for all $q \ge 5$. 
It has parameters $[6, 3, \log_q \left( \frac{q^4 - 1}{q-1} \right), \ge 2, 3]_q$.  
Since $\code$ has a minimum distance $3$, $\cL$ can correct any two symbol erasures. If those appear in different
basis vectors, the dimension loss error correcting capability matches that of the previously described subspace code. 
But the number of codewords in the code is strictly larger than that in $\cK$.
\end{example}

The increase in the number of codewords achieved through hybrid coding in the above scenario is negligible for large field orders. 
Furthermore, even these modest advantages are present only in cases when the symbol erasures (or errors) do not appear in bursts within
a few linearly independent vectors.

However, the advantage of the new construction is significantly more pronounced when the gap between $\ell$ and $m$ is large. 
This gap becomes of interest when the length of data transmitted in the network is much higher than the 
dimension of the subspaces used in coding, or in other words, when the min-cut of the network is small compared 
to the length of the coded vectors.  

\begin{example}
Take the code $\cK$ with parameters 
$[\ell + m, \ell, mk, 2 (\ell - k + 1)] = [12, 4, 16, 6]_q$. This code can correct up to and including two dimension 
losses and it contains $q^{16}$ codewords.

For comparison, take $W = \ff_q^{10}$ and consider the set $\cP = \cP(W, 4)$, where $|\cP| = {10 \brack 4}_q$. 
Let $\code$ be a $[12, 10, 3]_q$ GRS code, with $q \ge 11$. 
Define the mapping $\cE_\cL : W \rightarrow \code$ as before. 

The resulting code $\cL$ has parameters $[12, 4, \log_q \left( {10 \brack 4}_q \right), \ge 2, 3]_q$.  
Since $\code$ has a minimum distance $3$, $\cL$ can correct any two symbol erasures.  

The number of codewords in the code equals
\[
{10 \brack 4}_q = \frac{(q^{10} - 1)(q^9 - 1)(q^8 - 1)(q^7 - 1)}{(q^4 - 1)(q^3 - 1)(q^2 - 1)(q-1)} > q^{24}.
\]
This number is strictly larger than $4 q^{16}$ (for all $q \ge 11$),
which is an upper bound on the size of any $[12, 4, 16, 6]_q$ subspace code.  
\end{example}

\subsection*{Comparison of Dimension Losses and Symbol Erasures}

The examples described above motivate the following question: how many symbol erasures should be counted towards one dimension loss for the
case that the subspace and hybrid codes have the same number of codewords?
 
To arrive at the desired conclusion, we use an upper bound on the size of any constant-dimension 
subspace code, which was derived in~\cite{KK}. Therefore, our findings are also valid for the codes constructed in~\cite{KK, Skachek, Etzion-Silberstein}, 
as well as for any other possible subspace code. 

Throughout the section, we use $\hat{\cL}$ to denote an arbitrary $[n, \ell, \log(|\hat{\cL}|), 2 \tilde{D}]$ subspace code.     
Let us fix the values of the parameters $n$ and $\ell$. 
Recall that in the worst case, each symbol erasure can cause one dimension loss. 
We use the Singleton bound on the size of the code $\hat{\cL}$~\cite[Theorem 9]{KK}. Any such code 
is capable of correcting 
$\tilde{D}-1$ dimension losses, so in the worst case scenario, it can provably correct only up to $\tilde{D}-1$ symbol erasures. 
From~\cite[Theorem 9]{KK} we have
\[
|\hat{\cL}| \le {n - \tilde{D} + 1 \brack \ell - \tilde{D} + 1 }_q < 4 q^{(\ell - \tilde{D} + 1)(n - \ell)} \; . 
\]
In comparison, the number of codewords in the $[n, \ell, \log_q(|\cL|), 2D, d]$ code $\cL$ constructed in Section~\ref{sec:construct},
when $\Code$ is taken as in~\cite{KK}, is given by 
\[
|\cL| = q^{(\ell - D + 1)(n - \ell - d + 1)} \; .  
\]
In order to achieve the same erasure-correcting capability, we set $\tilde{D} - 1 = (D - 1) + (d - 1)$. The underlying
assumption is that $D-1$ symbol erasures are corrected as dimensional losses, while the remaining erasures are handled as simple erasures. 
We require that, for small $\epsilon > 0$,  
\[
(\ell - (\tilde{D} - 1))(n - \ell) + \epsilon < (\ell - (D - 1))(n - \ell - (d-1)) \; . 
\]
This is equivalent to 
\[
(\ell - (D - 1) - (d-1))(n - \ell) + \epsilon < (\ell - (D - 1))(n - \ell - (d-1)) \; ,  
\]
or 
\[
- (d-1)(n - \ell) + \epsilon < - (\ell - (D - 1)) (d-1) \; , 
\]
which reduces to 
\begin{equation}
( n - 2 \ell + (D-1)) (d - 1) > \epsilon \; .
\label{eq:comparison}
\end{equation}
The latter inequality holds for any choice of $D \ge 2$ and $d \ge 2$, 
when $n \ge 2\ell + \epsilon'$, for some small $\epsilon' > 0$. 
When the inequality~(\ref{eq:comparison}) is satisfied, hybrid codes correct more 
symbol erasures than any constant-dimension subspace code, designed to correct dimension errors only.  

Next, we consider maximizing the number of codewords in $\cL$ under the constraints that 
$(D - 1) + (d - 1) = \tilde{D} - 1$, and $D \ge 1$, $d \ge 1$, where $\tilde{D}$ is fixed and $D$, $d$ are allowed to vary. 
Recall that
\begin{equation}
\log_q(|\cL|) = (\ell - (D - 1))(n - \ell - (d-1)) \; . 
\label{eq:size-L}
\end{equation}
Let $x \define d - 1$ so that $D - 1 = s - x$, where $s = \tilde{D} - 1$ is a constant. 
We aim at maximizing the function 
\begin{equation}
\label{eq:object-function}
(\ell - s + x)(n - \ell - x) \; . 
\end{equation}
By taking the first derivative of the expression with respect to $x$ and by setting it to zero, we find that $x_{max} = \frac{n + s}{2} - \ell$. 
Therefore, the value of $d$ that maximizes the number of codewords equals
\[
d_{opt} = \frac{n + \tilde{D} + 1}{2} - \ell .  
\]
If $n \ge 2 \ell + \tilde{D} - 1$, then under the given constraints, 
the optimal value of $d$ equals $d_{opt} = \tilde{D}$, i.e. it is better to put all 
error-correcting capability on symbol erasure correction.  
\medskip


Consider the expression for the number of codewords in~(\ref{eq:size-L}). 
There are two types of subspace and symbol errors considered: dimension losses and symbol erasures. A 
combination of such errors is subsequently termed \emph{an error pattern}.

Assume that for a specific code $\cL$, correcting a dimension loss 
is on average equivalent to correcting $\csf$ symbol erasures, for some $\csf > 0$. 
We consider an optimal selection procedure for the parameters of $\cL$ for two different error patterns. 

If the error pattern consists of no dimension losses and $d-1$ symbol erasures, 
then~(\ref{eq:size-L}) becomes $\ell(n - \ell - (d-1))$. In comparison, 
if the error pattern consists of $D-1$ dimension losses and no symbol erasures, 
then~(\ref{eq:size-L}) becomes $(\ell - (D-1))(n - \ell)$. Since each 
dimension loss is on average equivalent to $\csf$ symbol erasures, we have 
\[
\csf \cdot (D-1) = d-1 \; , 
\]
and, so,
\[
(\ell - (d -1) / \csf )(n - \ell) = \ell(n - \ell - (d-1)) \; . 
\]
After applying some simple algebra, we obtain that
\[
\csf = (n - \ell)/\ell \; .
\]
Therefore, vaguely speaking, it is as hard to correct 
one dimension loss as it is to correct $(n - \ell)/\ell$ symbol erasures. 

\subsection{Hybrid Codes versus Coding Vectors}

We next compare hybrid codes with coding schemes used in randomized network coding~\cite{Ho}. 
In the latter approach, the information is presented by vectors rather than vector spaces. 
Special headers are appended to the 
vectors in order to keep track of which linear combinations were applied to the original packets. 
As a result, the coding vectors have an overhead of at least $\ell \cdot \log_2 (| \ff_q|)$ bits, where $\ell$ is the 
number of transmitted packets. This overhead becomes of less significance if the length of the data 
payload is large. 

When ``lifted'' subspace codes, such as~\cite{SKK}, 
are employed over the regular operator channel, 
the vector subspace is usually transmitted by using its basis vectors. Each such basis 
vector has a unique leading bit having value ``$1$''. The collection of such bits can be viewed as 
a header, used in coding vectors. When linear combinations are received at the destination, 
the original vectors can be recovered by applying the inverse linear transformation. Therefore, 
the coding schemes based on codes in~\cite{SKK} can also be viewed as a scheme based on coding vectors. 
Since the codes in~\cite{SKK} are asymptotically optimal with respect to the Singleton bound~\cite[Theorem 9]{KK}, 
we conclude that schemes based on coding vectors are asymptotically optimal. 
(See a related discussion in~\cite{Jafari}.)

Next, consider a hybrid code $\cL$ constructed as in Section~\ref{sec:construct}, where the matrix $\bldG$ is a \emph{systematic} generating matrix of a GRS code. 
Then, the transmitted basis vectors are the rows of $\bldG$, and so each vector has a unique leading bit ``$1$''. Therefore, the corresponding 
coding scheme may be viewed as a coding vector scheme with headers of size $(\ell+m) \cdot \log_2(|\ff_q|)$. By Corollary~\ref{cor:singleton-optimal}, 
this scheme is asymptotically optimal with respect to the corresponding
Singleton bound. Hence, asymptotically optimal hybrid subspace codes have a performance comparable to that of asymptotically optimal coding vector schemes. 

The advantages and drawbacks of the schemes based on coding vectors and on the subspace codes were discussed in~\cite{Jafari}. The analysis carried out there extends to hybrid codes, as they may be constructed around subspace codes. 

\section{Correcting Dimension Errors and Symbol Erasures} 
\label{sec:decoding-1} 

We proceed to present an efficient decoding procedure which handles  
dimension losses, dimension gains and symbol erasures. Note that the proposed 
decoding method may fail in the case that symbol errors are also present. This issue is discussed in more
detail in Section~\ref{sec:decoding-failure}.

As before, assume that $V \in \cL$ is transmitted over a noncoherent network.
Assume also that $U \in \cP(W_\cL, \ell')$, where $\ell'$ is not necessarily equal to $\ell$, was received. 

Let $U' \subseteq \ff_q^{n - \mu}$ denote the vector space $U \subseteq (\ff_q \cup \{ ? \})^{n}$, where all 
$\mu$ erased coordinates are deleted, $0 \le \mu \le d-1$. Similarly, let $\code'$ 
denote the code $\code$ where all $\mu$ coordinates erased in $U$ are deleted.  
We first compute $\tilde{U'} = \code' \cap U'$, the intersection of $U'$ with the subspace $\code'$. 
Assume that $\{ \bldgamma'_1, \bldgamma'_2, \cdots, \bldgamma'_{\ell''} \}$
are the basis vectors of $\tilde{U'}$ when the erased coordinates are marked as $?$, 
and so $\bldgamma'_i \in (\ff_q \cup \{ ? \})^{n}$. We apply the erasure-correcting GRS decoder $\decoder_{RS}$ of the code 
$\code$ to each $\bldgamma'_i$ so as to obtain $\bldgamma_i$. 
Let $\tilde{U} = \langle \bldgamma_1, \bldgamma_2, \cdots, \bldgamma_{\ell''} \rangle$.
We proceed by applying the inverse of the mapping $\cE_\cL$, denoted by $\cE_\cL^{-1}$, to $\tilde{U}$. 
The resulting subspace $\tilde{V}$ is a subspace of $W$, on which we now run the decoder for the code $\Code$. 

The algorithm described above is summarized in Figure~\ref{fig:alg-dimensions}. 

\begin{figure}[ht]
\makebox[0in]{}\hrulefill\makebox[0in]{}
\begin{description}
\item[\bf Input:] $\,$  $U \subseteq (\ff_q \cup \{ ? \})^{n}$. 
\settowidth{\Initlabel}{\textbf{Input:}}
\item[\bf Let] $U'$ be the space $U$, where all $\mu$ erased coordinates are deleted. 
\item[\bf Let] $\code'$ be the code $\code$, where all $\mu$ coordinates erased in $U$ are deleted.  
\item[\bf Let] $\tilde{U'} = \code' \cap U'$. 
\item[\bf Denote]  $\;\;$ $\tilde{U'} = \langle \bldgamma'_1, \bldgamma'_2, \cdots, \bldgamma'_{\ell''} \rangle$. 
\item[\bf For] $i = 1, 2 , \cdots, \ell''$ {\bf let} $\bldgamma_i = \cD_{RS} ({\bldgamma'}_i)$ \; . 
\item[\bf Let] $\tilde{U} = \langle \bldgamma_1, \bldgamma_2, \cdots, \bldgamma_{\ell''} \rangle$. 
\item[\bf Let] $\tilde{V} = \cE_\cL^{-1}( \tilde{U} )$.
\item[\bf Let] $V_0 = \cD_{\Code}( \tilde{V} )$.
\item[\bf Output:]  $\;\;$ $V_0$.
\end{description}
\makebox[0in]{}\hrulefill\makebox[0in]{}
\caption{Decoder for dimension errors.}
\label{fig:alg-dimensions}
\end{figure}	

This decoder can correct any combination of $\Theta$ dimension losses and $\Omega$ dimension gains such that $\Theta + \Omega \le D-1$, 
and at most $d-1$ symbol erasures. This is proved in the following theorem. 
%

\begin{theorem} 
The decoder in Figure~\ref{fig:alg-dimensions} can correct any error pattern of up to $D-1$ dimension errors
and up to $d-1$ symbol erasures in $\cL$. 
\end{theorem}

\begin{proof}
Suppose that $V \in \cL$ is transmitted through the operator channel and that
$U \subseteq (\ff_q \cup \{ ? \})^{n}$ is received, 
where $\mu$ symbols erasures and $\Theta + \Omega$
dimension errors have occurred, such that $\mu \le d-1$ and $\Theta + \Omega \le D-1$. 

As before, we assume that dimension errors have occurred first, followed by symbol erasures. More specifically, 
let $S = \{ j_1, j_2, \cdots, j_{\mu} \} \subseteq [n]$ be the set of erased coordinates in $U$, and let $W_\cL = W_S \bigodot \langle \blde_{j_1}, \blde_{j_2}, \cdots, \blde_{j_{\mu}} \rangle$. 
Then, 
\[
U_1 = \cH_k(V) \oplus E  \qquad \mbox{ and } \qquad U' = U_1 |_{W_S} \; ,
\]
where $\dim(V \cap U_1) = k$, $\dim(U_1) = \ell'$, and $(\ell - k) + (\ell' - k) = \Theta + \Omega$. 
In other words, first $U_1$ is obtained from $V$ by applying only dimension errors (losses and gains). 
Then, $U$ is obtained from $U_1$ by erasing coordinates in $S$. 
Note that here we assume that vectors in $U$ contain entries marked with `?', and so $U \subseteq (\ff_q \cup \{ ? \})^{n}$. By contrast, $U' \subseteq \ff_q^{n-\mu}$ is obtained by removing those erased entries from all vectors in $U$ (or, equivalently, from vectors in $U_1$). 

Denote $V' = V|_{W_S} \subseteq \code'$. Then, by Lemma~\ref{lemma:projection-distance}, 
\[
\Distance(V', U') \le \Distance(V, U_1) \le D - 1 \;  . 
\]

We have $\dim(V) = \dim(V')$, $\dim(\tilde{U}') \le \dim(U')$. 
Recall that $\tilde{U}' = \code' \cap U'$. Therefore, since $V' \subseteq \code'$, 
we have
\begin{multline*}
\dim(V' \cap \tilde{U'}) = \dim((U' \cap \code') \cap V') \\
= \dim(U' \cap (\code' \cap V')) = \dim(U' \cap V') \; . 
\end{multline*}
We consequently obtain
\begin{eqnarray*}
\Distance(V',\tilde{U}') & = & \dim(V') + \dim(\tilde{U}') - 2 \dim(\tilde{U}' \cap V') \\
& \le & \dim(V') + \dim(U') - 2 \dim({U}' \cap V') \\
& = & \Distance(V',U') \\ 
& \le & D  - 1 \; . 
\end{eqnarray*}

Observe that from Corollary~\ref{cor:preimage}, it follows that
$\dim(V) = \dim(V')$ and $\dim(\tilde{U}) = \dim(\tilde{U}')$.
Moreover, $\tilde{U}' \cap V'$ can be obtained from $\tilde{U} \cap V$
by $\mu$ symbol erasures. Then, according to Corollary~\ref{cor:preimage}, 
$\dim(\tilde{U}' \cap V') = \dim(\tilde{U} \cap V)$. 
We conclude that $\Distance(V, \tilde{U}) = \Distance(V',\tilde{U}') \le D - 1$.  

Finally, due to~(\ref{eq:mapping-1}) and~(\ref{eq:mapping-2}), we have that 
$\Distance(\cE_\cL^{-1}( V ), \cE_\cL^{-1}( \tilde{U} )) \le D - 1$. 
Therefore, the decoder $\decoder_\Code$ for the code $\Code$ is able to recover $\cE_\cL^{-1}( V )$, as claimed. 
\end{proof}
\medskip

We now turn to estimating the time complexity of hybrid decoding. 
The algorithms in Figure~\ref{fig:alg-dimensions} consists of the following main steps: 
\begin{itemize}
\item
\textbf{Computation of the vector space $\tilde{U'} = \code' \cap U'$.}\\
Observe that $\dim(\code') = \ell + m$ and $\dim(U') \le \ell'$, where $\ell' \le \ell + D \le 2 \ell$, since otherwise the decoder
cannot correct $D$ dimension errors. 
Therefore, this computation can be done by solving 
a system of at most $n$ equations over $\ff_q$ with $\ell + m + \ell'$ unknowns.   
By using Gaussian eliminations, this can be done in time $O( (\ell + m + \ell') n^2) = O( (\ell + m) n^2)$. 
\item 
\textbf{$\ell''$ applications of the decoder $\decoder_{RS} ( \cdot )$}.\\
Note that $\ell'' \le \ell' \le 2 \ell$. This requires 
$O(\ell n \log n)$ operations over $\ff_q$.
\item
\textbf{One application of the mapping $\tilde{V} = \cE_\cL^{-1}( \tilde{U} )$.}\\ 
As before, this step is equivalent to multiplying an $\ell'' \times n$ matrix representing the basis of $\tilde{U}$ by an $n \times (\ell + m)$ 
transformation matrix representing the mapping $\cE_\cL^{-1}( \cdot )$. 
This step requires $O\left( \ell'' (\ell + m) n \right) = O\left( \ell (\ell + m) n \right)$
operations over $\ff_q$.
\item
\textbf{One application of the decoder $\decoder_\Code( \cdot)$.} This takes $O( D (\ell + m)^3)$ operations over $\ff_q$ (see~\cite[Chapter 5]{Khaleghi}). 
\end{itemize}
By summing up all the quantities we arrive at an expression for the total complexity of the presented algorithm of the form
\begin{multline*}
O \left( (\ell+m) n^2 + \ell n \log n + \ell (\ell + m) n + D (\ell + m)^3 \right) \\
\le O \left( (\ell + m) n^2 + D (\ell + m)^3  \right) 
\end{multline*}
operations over $\ff_q$.

We note that the most time-consuming step in the decoding process for various choices of the parameters is decoding of a constant-dimension subspace code,
which requires $O( D (\ell + m)^3)$ operations over $\ff_q$. However, if the error pattern in a specific network 
contains a large number of symbol erasures, we can design hybrid codes such that $D$ is fairly small (say, some 
small constant). This reduces the complexity of the overall decoder, which represents another advantage of hybrid codes over classical subspace codes. 

\section{Correcting Dimensions Losses and Symbol Errors}
\label{sec:decoding-2}

We describe next how to use the code $\cL$ defined in Section~\ref{sec:construct}
for correction of error patterns that consist of dimension losses, symbol erasures and symbol substitutions. 
More specifically, we show that the code $\cL$ is capable of correcting any error pattern of up to $\Theta$ dimension losses,  
$\rho$ symbol errors and $\mu$ symbol erasures, whenever $\Theta \le D-1$ and $2 \rho + \mu \le d-1$. 
However, we note that if, in addition to dimension losses, one also encounters dimension gains, 
the decoder for the code $\cL$ may fail. This issue is elaborated on in Section~\ref{sec:decoding-failure}.

\subsection{Decoding}

Henceforth, we assume that $V \in \cL$ is transmitted  over a noncoherent network and 
that $U \subseteq \ff_q^{n-\mu}$ of dimension $\ell'$, where $\ell'$ is not necessarily equal to $\ell$, 
is received. 

Suppose that $\{ \bldgamma_1, \bldgamma_2, \cdots, \bldgamma_{\ell'} \}$, $\bldgamma_i \in \ff_q^{n-\mu}$, 
are the basis vectors in $U$. We can also view the vectors $\bldgamma_i$ as vectors in $(\ff_q \cup \{ ? \})^{n}$. 
We apply the GRS decoder $\decoder_{RS}$ for the code $\code$ on all these vectors.
This decoder produces the vectors $\{ \bldbeta_1, \bldbeta_2, \cdots, \bldbeta_{\ell'} \} \in \code$. 
We denote by $\tilde{U}$ the span of these vectors. Then, we apply the inverse of the mapping $\cE_\cL$, 
denoted by $\cE_\cL^{-1}$, to $\tilde{U}$. The resulting subspace is a subspace of $W$, on which the decoder
for the code $\Code$ is applied. 

The decoding algorithm is summarized in Figure~\ref{fig:alg-general}.
 
\begin{figure}[ht]
\makebox[0in]{}\hrulefill\makebox[0in]{}
\begin{description}
\item[\bf Input:] $\,$  $U = \langle \bldgamma_1, \bldgamma_2, \cdots, \bldgamma_{\ell'} \rangle$, $\bldgamma_i \in (\ff_q \cup \{ ? \})^{n}$. 
\settowidth{\Initlabel}{\textbf{Input:}}
\item[\bf For] $i = 1, 2 , \cdots, \ell'$ {\bf let} $\bldbeta_i = \cD_{RS} ({\bldgamma}_i)$ \; . 
\item[\bf Let] $\tilde{U} = \langle \bldbeta_1, \bldbeta_2, \cdots, \bldbeta_{\ell'} \rangle$.
\item[\bf Let] $\tilde{V} = \cE_\cL^{-1}( \tilde{U} )$.
\item[\bf Let] $V_0 = \cD_{\Code}( \tilde{V} )$.
\item[\bf Output:] $\;\;$ $V_0$.
\end{description}
\makebox[0in]{}\hrulefill\makebox[0in]{}
\caption{Decoder for symbol errors.}
\label{fig:alg-general}
\end{figure}	

\subsubsection*{Analysis of the Decoding Algorithm} 

We analyze next the algorithm in Figure~\ref{fig:alg-general}. 
The main result of this section is the following theorem. 

\begin{theorem} 
The decoder in Figure~\ref{fig:alg-general} can correct any error 
pattern in $\cL$ which consists of $\Theta$ dimension losses, $\rho$ symbol errors and $\mu$ symbol erasures, whenever $\Theta \le D-1$ and $2 \rho + \mu \le d-1$. 
\end{theorem}

\begin{proof}
Suppose that $V = \langle \bldv_1, \bldv_2, \cdots, \bldv_\ell \rangle \in \cL$ 
is transmitted through an operator channel, and that
$U = \langle \bldgamma_1, \bldgamma_2, \cdots, \bldgamma_{\ell'} \rangle$, 
$\bldgamma_i \in (\ff_q \cup \{ ? \})^{n}$, is obtained by the receiver. Assume that 
$\Theta$ dimension losses, $\rho$ symbol errors and $\mu$ symbol erasures have occurred. 

Let $\bldgamma_i$ be an arbitrary received vector, $\bldgamma_i \in U$. Then, $\bldgamma_i$ can be obtained from a unique 
vector $\tilde{\bldgamma}_i \in \ff_q^n$ by at most $\rho$ symbol errors and 
at most $\mu$ symbol erasures, where  
\[
\tilde{\bldgamma}_i = \sum_{j=1}^\ell a_j \bldv_j \; , 
\] 
and $a_j \in \ff_q$, $j = 1, 2, \cdots, \ell$. 
Since $\tilde{\bldgamma}_i$ is a linear 
combination of vectors in $V$, it follows that $\tilde{\bldgamma}_i \in \code$. 
Therefore, the decoder $\decoder_{RS}$ is able to recover $\tilde{\bldgamma}_i$ from $\bldgamma_i$. 
By using the structure of the algorithm, we conclude that $\bldbeta_i = \tilde{\bldgamma}_i$, and so 
$\tilde{U} = \langle \bldbeta_1, \bldbeta_2, \cdots, \bldbeta_{\ell'} \rangle$
is a subspace of $V$. 

Since $\Theta$ dimension losses occurred, $\Theta \le D-1$, 
$\dim(V) - \dim(\tilde{U}) \le \Theta$ and $\Distance(V, \tilde{U}) \le D - 1$.  
Due to~(\ref{eq:mapping-1}) and~(\ref{eq:mapping-2}), we have that 
$\Distance(\cE_\cL^{-1}( V ), \cE_\cL^{-1}( \tilde{U} )) \le D - 1$. 
Therefore, the decoder $\decoder_\Code$ is able to recover $\cE_\cL^{-1}( V )$ from $\cE_\cL^{-1}( \tilde{U} )$, 
as claimed. 
\end{proof}

\subsubsection*{Decoding Time Complexity}

The decoding algorithm consists of the following computational steps: 
\begin{itemize}
\item
\textbf{$\ell'$ applications of a Reed-Solomon decoder, for codes of length $n = \ell + m +d -1$.}\\ 
By using Berlekamp-Massey type decoders, each decoding round can be performed with $O\left( n \log n \right)$ 
operations over $\ff_q$. Thus, this step has a total complexity of $O\left( \ell' n \log n) \right)$. 
\item
\textbf{One application of the mapping $\tilde{V} = \cE_\cL^{-1}( \tilde{U} )$.}\\
First, we have to find a basis for $\tilde{U}$. Gaussian elimination requires 
$O(\ell'^2 n)$ operations over $\ff_q$. 
The mapping $\cE_\cL^{-1}( \tilde{U} )$ is equivalent to multiplying an $\ell' \times n$ matrix representing the basis of $\tilde{U}$ by an 
$n \times (\ell + m)$ transformation matrix representing the mapping $\cE_\cL^{-1}( \cdot )$. The computation of this transformation 
matrix is done only once in the preprocessing step, and so we may assume that this matrix is known. 
We hence conclude that this step takes $O\left( \ell' (\ell + m) n + \ell'^2 n \right)$
operations over $\ff_q$.
\item
\textbf{One application of the decoder $\decoder_\Code( \cdot)$.}\\ 
This takes $O( D (\ell + m)^3)$ operations 
over $\ff_q$ (see~\cite[Chapter 5]{Khaleghi}). 
\end{itemize}
The total complexity of the presented algorithm equals
\begin{multline*}
O \left( \ell' n \log n + \ell' (\ell + m) n + \ell'^2 n + D (\ell + m)^3 \right) \\ 
\le O \left( D n^3 + \ell' n^2 + \ell'^2 n \right) 
\; 
\end{multline*}
operations over $\ff_q$.

The number of operations depends on the dimension of the received subspace, $\ell'$. It would hence be desirable to derive an upper bound on $\ell'$. 
However, since each linearly independent vector can carry a different pattern of symbol errors, the resulting dimension of $U$, $\ell'$, may be rather large. 
However, if we assume that each link to the receiver carries only one vector, $\ell'$ can be bounded from above 
by the in-degree of the receiver. Note that the same issue arises in the context of classical subspace coding, 
although it was not previously addressed in the literature.

\section{Four Types of Errors and Decoding Failure}
\label{sec:decoding-failure}

\subsection{Failure Example}

The following example illustrates that the decoder in Figure~\ref{fig:alg-general} may fail in the presence of \emph{both symbol errors and dimension gains}.  

Let $\{ \blde_1, \blde_2, \cdots, \blde_6 \} \subseteq \ff_q^6$, $q \ge 8$, be a standard basis. Let $\ell = 3$,  
and let $\Code \subseteq \cP(\ff_q^6)$ be a subspace code with $2 D  = 6$. The code $\Code$ is able 
to correct up to and including two dimension losses and/or gains. 
Additionally, let 
$\{ \bldu_1, \bldu_2, \cdots, \bldu_6 \} \in \ff_q^8$ be a basis of a $[8, 6, 3]_q$ GRS code $\code$. 
The code $\code$ is able to correct one symbol error. Assume, without loss of generality, that 
$\bldu_5 = (x_1, x_2, x_3, 0, \cdots, 0) \in \code$ is a codeword of a minimal weight in $\code$.  
 
Assume that the sender wants to transmit the space $Z = \langle \blde_1, \blde_2, \blde_3 \rangle$ 
to the receiver. According to the algorithm, the sender encodes this space as 
$V = \cE_\cL(Z) = \langle \bldu_1, \bldu_2, \bldu_3 \rangle$, and sends the vectors $\bldu_1$, $\bldu_2$, $\bldu_3$ 
through the network. Assume that the vector $\bldu_3$ is removed, and the 
erroneous vector $\bldz = \bldu_4 + (x_1, 0, \cdots, 0), x_1 \neq 0$ is injected instead. 
At this point, the corresponding vector space under transmission 
is $\langle \bldu_1, \bldu_2, \bldz \rangle$. 
Then, it is plausible that $\bldu_1$, $\bldu_2$ and $\bldz$ propagate further through the network due to network coding. 
To this end, assume that the receiver receives the following linear combinations, $\bldu_1 + \bldz$ and $\bldu_2 + \bldz$. 
Assume also that during the last transmission, the vector $\bldz$ is subjected to a symbol error, 
resulting in $\bldz' = \bldu_4 + (x_1, x_2, 0, \cdots, 0), x_2 \neq 0$. 

The receiver applies the decoder $\cD_{RS}$ on these three vectors, resulting in 
\begin{eqnarray*}
\cD_{RS}(\bldu_1 + \bldz) = \bldu_1 + \bldu_4 \; ; \\
\cD_{RS}(\bldu_2 + \bldz) = \bldu_2 + \bldu_4 \; ; \\
\cD_{RS}(\bldz') = \bldu_4 + \bldu_5 \; . 
\end{eqnarray*}
The subspace received at the destination is 
\[
\tilde{U} = \langle \bldu_1 + \bldu_4, \bldu_2 + \bldu_4, \bldu_4 + \bldu_5 \rangle \; , 
\]
and so the corresponding pre-image under $\cE_\cL$ is given by
\[
\tilde{V} = \langle \blde_1 + \blde_4, \blde_2 + \blde_4, \blde_4 + \blde_5 \rangle \; . 
\]
Observe that $\dim(Z \cap \tilde{V}) = 1$ and that $\blde_1 + \blde_2 \in Z \cap \tilde{V}$, so that the 
subspace distance between $V$ and $\tilde{V}$ is four. Therefore, the subspace decoder $\cD_{\Code}$ may fail when decoding $Z$ from $\tilde{V}$. 
This situation is illustrated in Figure~\ref{fig:confusion}.  

\begin{figure}[ht]
\begin{center}
  \includegraphics[width=0.45\textwidth]{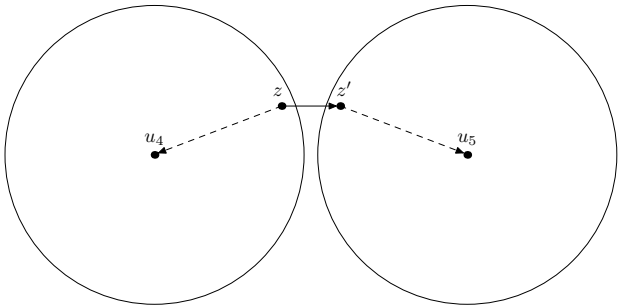}
\end{center}
\caption{Situation when the decoder fails. The dimension error vector $\bldz$ should be decoded into $\bldu_4 \in \code$. 
However, a symbol error changes $\bldz$ into $\bldz'$, which is decoded into $\bldu_5 \in \code$. This ambiguity 
increases the dimension of the error space, and causes decoder failure.}
\label{fig:confusion}  
\end{figure}

\subsection{Decoding Strategies for Dimension Gains and Symbol Errors}

To illustrate the difficulty of performing combined symbol and dimension gain error decoding of the code $\cL$, below we discuss some 
alternative decoding strategies. We mention why these strategies, when applied to the problem at hand, do not work. 

\begin{description}
\item[\bf Gaussian eliminations on the orthogonal space.]
$\,$
\newline
Assume that the vector space $V$ is transmitted and $U$ is obtained by a combination of symbol and dimension errors, including dimension gains. 
Then, $U$ can be represented as $U = \cH_k(V) \oplus E$, where $E$ is some error space. 
If there were no symbol errors, the dimension of $E$ would equal the number of dimension gains, which is small. 
However, if symbol errors are present, $E$ takes a more involved form. 

One can try to represent the space $E$ in a particular basis, for example one in which the symbol errors have low weight. 
Ideally, each symbol error would correspond to a vector of weight one in that space. 
If one could accomplish this task, then it may be possible to find all the low weight vectors and remove them, or to puncture the corresponding coordinates. After such a procedure, 
one would ideally be left with only dimension gain vectors. 

A particular difficulty in this scenario is that there are too many different bases for $E$, and it is not immediately clear which basis should be selected. 
And, while in the right basis the symbol error vector will have weight one, in most of the other bases this weight will be large. 
Moreover, the space $E$ can be viewed as a dual code of $\cL$. However, then the problem of finding low-weight vectors becomes similar to the problem of 
finding the smallest weight codeword in the dual code, which is known to be NP hard. Therefore, it is likely that finding the right basis in $E$ is difficult, too. 

\item[\bf Using list-decoding for RS codes.]
$\,$
\newline
One can think about using list-decoding for the code $\code$. Since the covering radius of RS codes is $d-1$, it may happen that the dimension 
error transforms the codeword into a vector at distance $d-1$ from any codeword. Then, even a single symbol error can move this codeword to a different ball of radius $d-1$ around a codeword, 
similarly to the situation depicted in Figure~\ref{fig:confusion}. Since list-decoding can correct only less than $d$ errors, list-decoding cannot recover the original codeword.  

The second problem associated with list decoding is as follows. Even if one could construct a polynomial-size list of all possible codewords before the dimension error took place, 
there would be a different list for each received vector. Since there could be as many as $\ell$ different lists, an exhaustive approach for choosing the right codewords from all the lists 
can require a time exponential in $\ell$. 
\end{description}


\section{Conclusion}

We introduced a new class of subspace codes capable of correcting both dimension errors and symbol errors, termed hybrid codes. For these codes,
we derived upper bounds on the size of the codes and presented an asymptotically constant-optimal concatenated code design method. 
We presented polynomial-time decoding algorithms which are capable of correcting the following error
patterns: 
\begin{itemize}
\item
Dimension losses/gains and symbol erasures; 
\item
Dimension losses and symbol erasures/errors. 
\end{itemize}
We also discussed correction of error patterns that consist of all four types of errors: dimension losses/gains and symbol erasures/errors. 
As we illustrated by an example, the corresponding task is difficult, and is left as an open problem. 

\section{Acknowledgements}

The authors are grateful to Danilo Silva for providing useful and insightful comments about the work
in the manuscript, the reviewers for many comments that improved the presentation of the work, and
to the Associate Editor, Christina Fragouli, for insightful suggestions and for handling the manuscript. 


%

\end{document}